\documentclass[cmodern]{qjmam}

\usepackage{amsmath}
\usepackage{amsfonts}
\usepackage{amssymb}
\usepackage{amsthm}
\usepackage{graphicx}
\usepackage{bm}
\usepackage{booktabs}

\usepackage{url}
\usepackage{xcolor}

\startpage{1}
\yr{2016}
\vol{52}
\issue{1}

\newtheorem{definition}{Definition}[section]

\newtheorem{Theorem}[definition]{Theorem}

\DeclareMathAlphabet\mathbit
    \encodingdefault\rmdefault\bfdefault\itdefault
\DeclareOldFontCommand{\bi}{\normalfont\bfseries\itshape}{\mathbit}

\newcommand{\be}{\begin{equation}}
\newcommand{\ee}{\end{equation}}

\def\fakebold#1{\relax\ifvmode\leavevmode\fi%
\ifmmode%
\setbox0=\hbox{$#1$}%
\else%
\setbox0=\hbox{#1}%
\fi%
\kern-.02em\copy0 \kern-\wd0%
\kern .04em\copy0 \kern-\wd0%
\kern-.0125em\raise.02em\box0%
}%

\begin{document}

\title[{The Riemann problem of 1D elastodynamics}] {Analytical solution to the Riemann problem of 1D elastodynamics with general constitutive laws}

\author[H. Berjamin \etal] {H. Berjamin \and B. Lombard}
\address{LMA, CNRS, UPR~7051, Aix-Marseille~Univ., Centrale~Marseille, 13453~Marseille~Cedex~13, France}

\extraauthor{G. Chiavassa}
\extraaddress{Centrale~Marseille, CNRS, Aix-Marseille Univ., M2P2~UMR~7340, 13451~Marseille~Cedex~20, France}

\extraauthor{N. Favrie}
\extraaddress{Aix-Marseille Univ., UMR~CNRS~7343, IUSTI, Polytech~Marseille, 13453~Marseille~Cedex~13, France}

\received{\recd DD Month YYYY. \revd DD Month YYYY}

\maketitle

\eqnobysec

\begin{abstract} 
	Under the hypothesis of small deformations, the equations of 1D elastodynamics write as a $2\times 2$ hyperbolic system of conservation laws. Here, we study the Riemann problem for convex and nonconvex constitutive laws. In the convex case, the solution can include shock waves or rarefaction waves. In the nonconvex case, compound waves must also be considered. In both convex and nonconvex cases, a new existence criterion for the initial velocity jump is obtained. Also, admissibility regions are determined. Lastly, analytical solutions are completely detailed for various constitutive laws (hyperbola, tanh and polynomial), and reference test cases are proposed.
\end{abstract}

	\section{Introduction}\label{sec:Intro}	
	
	The behavior of elastic media is characterized by the stress-strain relationship, or constitutive law. For many materials such as rocks, soil, concrete and ceramics, it appears to be strongly nonlinear \cite{guyer99}, in the sense that nonlinearity occurs even when the deformations are small. Extensive acoustic experiments have been carried out on sandstones \cite{guyer99,johnson96a,johnson96b,abeele96} and on polycristalline zinc \cite{nazarov00}. In these experiments, the sample is a rod of material, which is resonating longitudinally.
	
	For this kind of experiments, one-dimensional geometries are often considered. Moreover, the small deformations hypothesis is commonly assumed. Therefore, the stress $\sigma$ is a function of the axial strain $\varepsilon$, for example a hyperbola, a hyperbolic tangent (tanh), or a polynomial function. Known as Landau's law \cite{landaulifschitz59}, the latter is widely used in the community of nondestructive testing \cite{mccall94b,lombardWM15}.
	
	Under these assumptions, elastodynamics write as a $2 \times 2$ hyperbolic system of conservation laws. For general initial data, no analytical solution is known when $\sigma(\varepsilon)$ is nonlinear. 
	Analytical solutions can be obtained in the particular case of piecewise constant initial data having a single discontinuity, i.e. the Riemann problem. Computing the solution to the Riemann problem is of major importance to get a theoretical insight on the wave phenomena, but also for validating numerical methods.
	
	When $\sigma$ is a convex or a concave function of $\varepsilon$, one can apply the techniques presented in \cite{godlewski96} for the $p$-system of barotropic gas dynamics. In this reference book, a condition which ensures the existence of the solution is presented. This condition has been omitted in \cite{meurer02}, in the case of the quadratic Landau's law. We prove that this kind of condition is obtained also in the case of elastodynamics, and that it involves also a restriction on the initial velocity jump. Furthermore, it is shown in \cite{godlewski96} how to predict the nature of the physically admissible solution in the case of the $p$-system. We present here how it can be applied to elastodynamics.
	
	When $\sigma$ has an inflexion point, it is neither convex nor concave. The physically admissible solution is much more complex than for the $p$-system, but the mathematics of nonconvex Riemann problems are well established \cite{wendroff72a,liu74,dafermos05}. It has been applied to elastodynamics in \cite{shearer95}, but with a negative Young's modulus, which is not physically relevant. Here, we state a condition which ensures the existence of the solution to the Riemann problem. Also, we show how to predict the nature of the physically admissible solution. Finally, we provide a systematic procedure to solve the Riemann problem analytically, whenever $\sigma$ has an inflexion point or not. In the case of Landau's law, an interactive application and a Matlab toolbox can be found at {\color{blue}\url{http://gchiavassa.perso.centrale-marseille.fr/RiemannElasto/}}.
	
	%-------------------------- Elasto --------------------------------
	
	\section{Preliminaries}\label{sec:Prelim}
	
	\subsection{Problem statement}\label{subsec:ProbStatement}
	
	Let us consider an homogeneous one-dimensional continuum. The Lagrangian representation of the displacement field is used. Under the assumption of small deformations, the mass density is constant. Therefore, it equals the density $\rho_0$ of the reference configuration. Elastodynamics write as a $2\times 2$ system:
	\begin{equation}
		\left\lbrace
		{\addtolength{\jot}{0.5em}
		\begin{aligned}
			&\frac{\partial \varepsilon}{\partial t} = \frac{\partial v}{\partial x}\, ,\\
			&\rho_0\, \frac{\partial v}{\partial t} = \frac{\partial}{\partial x}\,\sigma(\varepsilon)\, .
		\end{aligned}}
		\right.
		\label{SystHyp}
	\end{equation}
	If $u$ denotes the $x$-component of the displacement field, then $\varepsilon = \partial u/\partial x$ is the infinitesimal strain, and $v = \partial u/\partial t$ is the particle velocity. We assume that the stress $\sigma$ is a smooth function of $\varepsilon$, which is strictly increasing over an open interval $\left]\varepsilon_\text{\it inf},\varepsilon_\text{\it sup}\right[$ with $\varepsilon_\text{\it inf}<0$ and $\varepsilon_\text{\it sup}>0$. These bounds $\varepsilon_\text{\it inf}$ and $\varepsilon_\text{\it sup}$ can be finite or infinite. Also, no prestress is applied, i.e. $\sigma(0)=0$. When replacing $\varepsilon$ by the specific volume $v$, $-\sigma$ by the pressure $p$, $v$ by the particle velocity $u$ and $\rho_0$ by $1$ in (\ref{SystHyp}), the so-called ``$p$-system'' of gas dynamics is recovered~\cite{wendroff72a}.
		
	As a set of conservation equations, the system (\ref{SystHyp}) can be written in the form
	\begin{equation}
		\frac{\partial}{\partial t} \bm{U} + \frac{\partial}{\partial x} \bm{f}(\bm{U}) = \bm{0}\, ,
		\label{SystHypVect}
	\end{equation}
	where $\bm{U}=(\varepsilon,v)^\top$ and $\bm{f}(\bm{U})=-(v,\sigma(\varepsilon)/\rho_0)^\top$. The Riemann problem for this system is defined by the piecewise constant initial data
	\begin{equation}
		\bm{U}(x,0) =
		\left\lbrace\!
		\begin{array}{ll}
			\bm{U}_L &\mbox{if } x<0\, ,\\
			\bm{U}_R &\mbox{elsewhere}\, ,
		\end{array}\right.
		\label{SystCI}
	\end{equation}
	with $\bm{U}_L = (\varepsilon_L, v_L)^\top$ and $\bm{U}_R = (\varepsilon_R, v_R)^\top$. Solving (\ref{SystHypVect})-(\ref{SystCI}) is the goal of the next sections.
	
	\subsection{Characteristic fields}\label{subsec:CharFields}
	
	The Jacobian matrix of $\bm{f}$ in (\ref{SystHypVect}) is
	\begin{equation}
	\bm{f}'(\bm{U})
	=
	-\left(\!{\renewcommand{\arraystretch}{1}
		\begin{array}{cc}
		0 & 1 \\
		\sigma'(\varepsilon)/\rho_0 & 0
		\end{array}}\!\right)
	\label{SystHypJacob}
	\end{equation}
	with eigenvalues
	\begin{equation}
	\lambda_1(\bm{U}) = {-c}(\varepsilon),\qquad \lambda_2(\bm{U}) = c(\varepsilon) ,
	\label{SystHypValP}
	\end{equation}
	where
	\begin{equation}
	c(\varepsilon) = \sqrt{\sigma'(\varepsilon)/\rho_0}
	\label{SystHypValPC}
	\end{equation}
	is the speed of sound. The right eigenvectors $\bm{r}_p$ and left eigenvectors $\bm{l}_p$ satisfy ($p=1$ or $p=2$)
	\begin{equation}
	{\addtolength{\jot}{0.3em}
		\begin{aligned}
		&\bm{f}'(\bm{U})\, \bm{r}_p(\bm{U}) = \lambda_p(\bm{U})\, \bm{r}_p(\bm{U})\, ,\\
		&{\bm{l}_p}(\bm{U})^\top\, \bm{f}'(\bm{U}) = \lambda_p(\bm{U})\, {\bm{l}_p}(\bm{U})^\top .
		\end{aligned}}
	\label{EigVectors}
	\end{equation}
	They can be normalized in such a way that ${\bm{l}_p}(\bm{U})^\top \,\bm{r}_p(\bm{U}) = 1$. Thus,
	\begin{equation}
	{\addtolength{\jot}{0.3em}
		\begin{aligned}
		&\bm{r}_1(\bm{U}) =
		\left(\!{\renewcommand{\arraystretch}{1}
			\begin{array}{c}
			1\\
			c(\varepsilon)
			\end{array}}\!\right) ,
		\quad & &\bm{r}_2(\bm{U}) =
		\left(\!{\renewcommand{\arraystretch}{1}
			\begin{array}{c}
			1\\
			{-c}(\varepsilon)
			\end{array}}\!\right) , \\
		&\bm{l}_1(\bm{U}) = \frac{1}{2}
		\left(\!{\renewcommand{\arraystretch}{1}
			\begin{array}{c}
			1\\
			1/c(\varepsilon)
			\end{array}}\!\right) ,
		\quad & &\bm{l}_2(\bm{U}) = \frac{1}{2}
		\left(\!{\renewcommand{\arraystretch}{1}
			\begin{array}{c}
			1\\
			-1/c(\varepsilon)
			\end{array}}\!\right) .
		\end{aligned}}
	\label{SystHypVectP}
	\end{equation}
	If the eigenvalues $\lambda_p$ of a $2\times 2$ system of conservation laws are real and distinct over an open set $\Omega$ of $\mathbb{R}^2$, then the system is \emph{strictly hyperbolic} over $\Omega$ \cite{godlewski96}. Here, the system~(\ref{SystHyp}) is strictly hyperbolic if $\sigma$ is strictly increasing, i.e. $\Omega=\left]\varepsilon_\textit{inf},\varepsilon_\text{\it sup}\right[ \times\mathbb{R}$.
	
	If the $p$th characteristic field satisfies $\bm{\nabla} \lambda_p\cdot \bm{r}_p = 0$ for all states $\bm{U}$ in $\Omega$, then it is \emph{linearly degenerate}. Based on (\ref{SystHypValP}), linear degeneracy reduces to
	\begin{equation}
		\sigma(\varepsilon) = E\,\varepsilon\, ,
		\label{ElastoLin}
	\end{equation}
	where $E>0$ is the Young's modulus. Therefore, (\ref{ElastoLin}) corresponds to the classical case of linear elasticity~\cite{achenbach73}. When linear degeneracy is not satisfied, the classical case is obtained when $\bm{\nabla} \lambda_p\cdot \bm{r}_p \neq 0$ for all states $\bm{U}$ in $\Omega$. The $p$th characteristic field is then \emph{genuinely nonlinear}. Here, this is equivalent to state for all $\varepsilon$ in $\left]\varepsilon_\textit{inf}, \varepsilon_\text{\it sup}\right[$,
	\begin{equation}
		\sigma''(\varepsilon) \neq 0 .
		\label{NLGen}
	\end{equation}
	Therefore $\sigma$ is either a strictly convex function or a strictly concave function. In the case of linear elasticity (\ref{ElastoLin}), one can remark that $\sigma$ is still convex. A less classical case is when both $\bm{\nabla} \lambda_p\cdot \bm{r}_p = 0$ and $\bm{\nabla} \lambda_p\cdot \bm{r}_p \neq 0$ can occur over $\Omega$. This happens when $\sigma''$ has isolated zeros. $\sigma$ is therefore neither convex nor concave. In this study, we restrict ourselves to a single inflexion point $\varepsilon_0$ in $\left]\varepsilon_\textit{inf},\varepsilon_\text{\it sup}\right[$ such that
	\begin{equation}
		\sigma''(\varepsilon_0) = 0\, .
		\label{NLNGen}
	\end{equation}
	
	Three constitutive laws $\varepsilon\mapsto \sigma(\varepsilon)$ have been chosen for illustrations. They cover all the cases related to convexity or to the hyperbolicity domain. Among them, the polynomial Landau's law is widely used in the experimental literature \cite{johnson96a,abeele96}, and the physical parameters given in table~\ref{tab:Params} correspond to typical values in rocks.
	
	\begin{table}
		\caption{Physical parameters.\label{tab:Params}}
		\centering
		{\renewcommand{\arraystretch}{1}
		\renewcommand{\tabcolsep}{0.1cm}
		\begin{tabular}{ccccc}
			\toprule
			$\rho_0$ (kg.m$^{-3}$) & $E$ (GPa) & $d$ & $\beta$ & $\delta$ \\
			\midrule
			$2600$ & $10$ & $10^{-3}$ & $10^2$ & $10^6$ \\
			\bottomrule
		\end{tabular}}
	\end{table}
		
	\begin{figure}
		\begin{minipage}{0.5\linewidth}
			\centering
			(a)
				
			\includegraphics{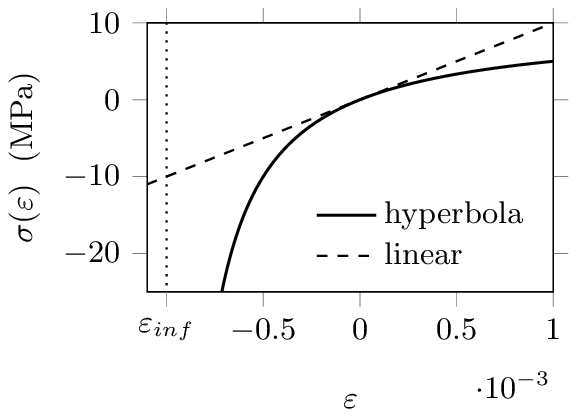}
		\end{minipage}
		\begin{minipage}{0.5\linewidth}
			\centering
			(b)
				
			\includegraphics{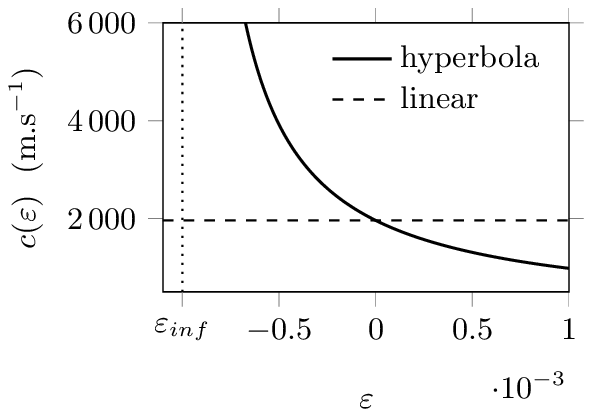}
		\end{minipage}
		\caption{(a) hyperbola constitutive law (\ref{BLawHyp}) and (b) speed of sound (\ref{CHyp}) compared to the linear case (\ref{ElastoLin}).\label{fig:Hyp}}
	\end{figure}
	
	\paragraph*{Model 1 {\mdseries (hyperbola)}.}
	This constitutive law writes
	\begin{equation}
		\sigma(\varepsilon) = \frac{E\, \varepsilon}{1+\varepsilon/d}\, ,
		\label{BLawHyp}
	\end{equation}
	where $d>0$. Here, $\left]\varepsilon_\textit{inf},\varepsilon_\text{\it sup}\right[ = \left]-d,+\infty\right[$. At the bound $\varepsilon_\textit{inf}$, $\sigma$ has a vertical asymptote. Figure~\ref{fig:Hyp} displays the law (\ref{BLawHyp}) and its sound speed
	\begin{equation}
		c(\varepsilon)=\frac{c_0}{1+\varepsilon/d}\, ,
		\label{CHyp}
	\end{equation}
	where
	\begin{equation}
		c_0 = \sqrt{E/\rho_0}
		\label{Clin}
	\end{equation}
	is the speed of sound in the linear case (\ref{ElastoLin}).	Since $\sigma''$ does not vanish over $\left]\varepsilon_\textit{inf},\varepsilon_\text{\it sup}\right[$, the characteristic fields are genuinely nonlinear (\ref{NLGen}).
		
	\begin{figure}
		\begin{minipage}{0.5\linewidth}
			\centering
			(a)
				
			\includegraphics{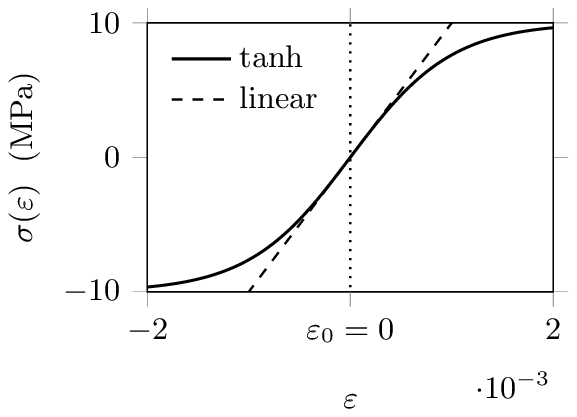}
		\end{minipage}
		\begin{minipage}{0.5\linewidth}
			\centering
			(b)
				
			\includegraphics{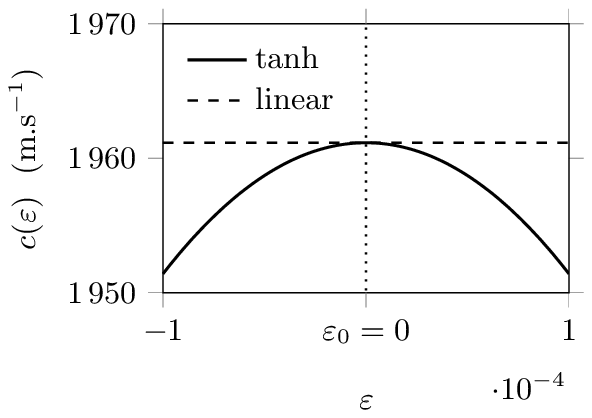}
		\end{minipage}
		\caption{(a) tanh constitutive law (\ref{BLawHypTan}) and (b) speed of sound (\ref{CHypTan}), zoom. \label{fig:HypTan}}
	\end{figure}
		
	\paragraph*{Model 2 {\mdseries (tanh)}.}
	This constitutive law writes
	\begin{equation}
		\sigma(\varepsilon) = E\, d\, \tanh(\varepsilon/d)\, ,
		\label{BLawHypTan}
	\end{equation}
	where $d>0$. Figure \ref{fig:HypTan} displays the law (\ref{BLawHypTan}) and its sound speed
	\begin{equation}
		c(\varepsilon)= \frac{c_0}{\cosh(\varepsilon/d)} \, .
		\label{CHypTan}
	\end{equation}
	Strict hyperbolicity is ensured for all $\varepsilon$ in $\mathbb{R}$. Among the constitutive laws considered here, the tanh is the only model with an unbounded hyperbolicity domain. However, $\sigma''(\varepsilon_0)=0$ at $\varepsilon_0=0$. Therefore, genuine nonlinearity is not satisfied (\ref{NLNGen}). At $\varepsilon_0$, the sound speed reaches its maximum $c(\varepsilon_0) = c_0$ (\ref{Clin}) .
		
	\paragraph*{Model 3 {\mdseries (Landau)}.}
	This constitutive law writes~\cite{landaulifschitz59}
	\begin{equation}
		\sigma(\varepsilon) = E\, \varepsilon\left(1-\beta\,\varepsilon-\delta\,\varepsilon^2\right) ,
		\label{BLawLandau}
	\end{equation}
	where $E$ is the Young's modulus and $(\beta,\delta)$ are positive. Figure~\ref{fig:Landau} represents the constitutive law (\ref{BLawLandau}) and its sound speed
	\begin{equation}
		c(\varepsilon)=c_0\, \sqrt{1-2\beta\,\varepsilon-3\delta\,\varepsilon^2}\, .
		\label{CLandau}
	\end{equation}
	In the particular case where the nonlinearity in (\ref{BLawLandau}) is quadratic ($\delta= 0$), the hyperbolicity domain is $\left]\varepsilon_\textit{inf},\varepsilon_\text{\it sup}\right[ = \left]{-\infty},1/2\beta\right[$. At the bound $\varepsilon_\textit{inf}$, $\sigma$ has a zero slope. A truncated Taylor expansion of the hyperbola law (\ref{BLawHyp}) at $\varepsilon=0$ recovers the quadratic Landau's law when replacing $\beta$ by $1/d$. Both laws $\varepsilon\mapsto \sigma(\varepsilon)$ are strictly concave, and their characteristic fields are genuinely nonlinear. When the nonlinearity is cubic ($\delta\neq 0$), hyperbolicity is satisfied when $\varepsilon$ belongs to
	\begin{equation}
		\left] \varepsilon_\textit{inf}, \varepsilon_\textit{sup} \right[
		= \left] \frac{1}{\beta-\sqrt{\beta^2+3\delta}}, \frac{1}{\beta+\sqrt{\beta^2+3\delta}} \right[ .
		\label{OmegaLandau}
	\end{equation}
	At the bounds $\varepsilon_\textit{inf}$ and $\varepsilon_\textit{sup}$, $\sigma$ has a zero slope. A truncated Taylor expansion of the tanh model (\ref{BLawHypTan}) at $\varepsilon=0$ recovers Landau's law when replacing $\beta$ by $0$ and $\delta$ by $1/d^2$. Both laws $\varepsilon\mapsto \sigma(\varepsilon)$ have an inflexion point $\varepsilon_0$, and their characteristic fields are not genuinely nonlinear (\ref{NLNGen}). Here, $\sigma''(\varepsilon_0)=0$ at $\varepsilon_0 = -\beta/3\delta$, where the sound speed reaches its maximum value
	\begin{equation}
		c(\varepsilon_0) = c_0\,\sqrt{1+\frac{\beta^2}{3\delta}}\, > c_0\, .
		\label{CLandauMax}
	\end{equation}
	
	\begin{figure}
		\begin{minipage}{0.5\linewidth}
			\centering
			(a)
			
			\includegraphics{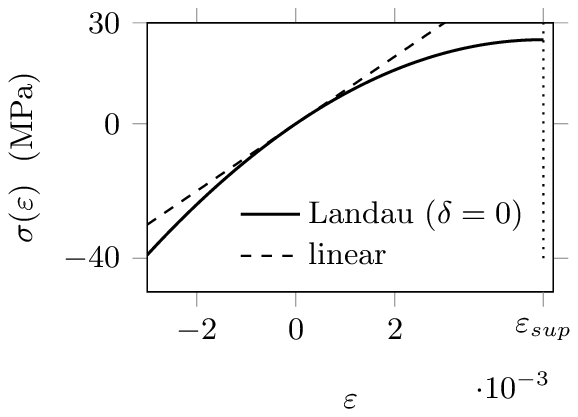}
		\end{minipage}
		\begin{minipage}{0.5\linewidth}
			\centering
			(b)
			
			\includegraphics{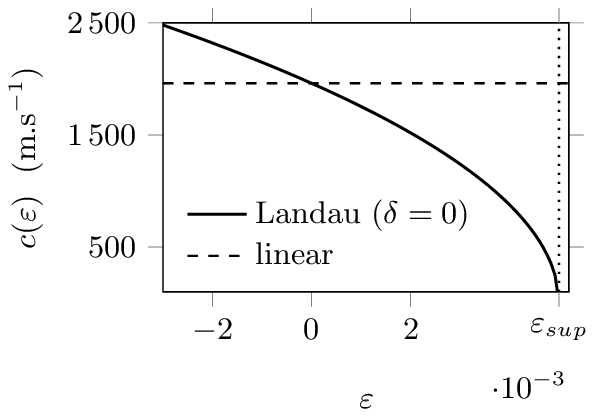}
		\end{minipage}
		\vspace{1em}
		
		\begin{minipage}{0.5\linewidth}
			\centering
			(c)
			
			\includegraphics{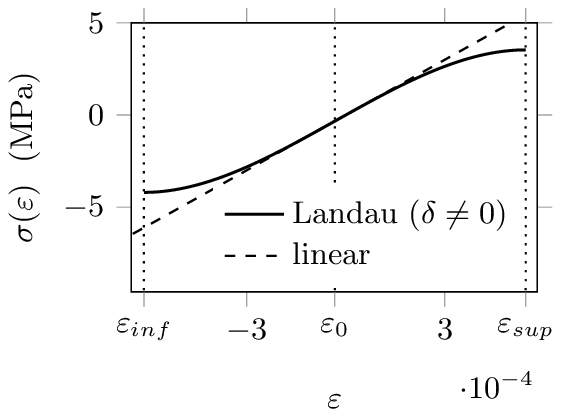}
		\end{minipage}
		\begin{minipage}{0.5\linewidth}
			\centering
			(d)
			
			\includegraphics{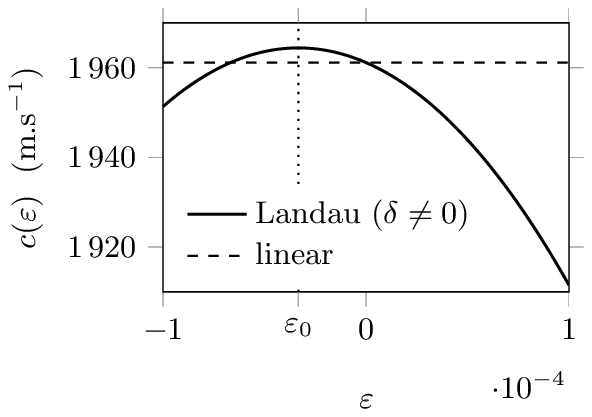}
		\end{minipage}
		\caption{(a) Landau's law (\ref{BLawLandau}) with a quadratic nonlinearity ($\delta = 0$) and (b) speed of sound (\ref{CLandau}). (c)-(d) Idem with a cubic nonlinearity ($\delta \neq 0$), zoom.\label{fig:Landau}}
	\end{figure}
	
	In the linearly degenerate case (\ref{ElastoLin}), the solution to the Riemann problem (\ref{SystHypVect})-(\ref{SystCI}) consists of two contact discontinuities propagating at speed $\mp c_0$. In the genuinely nonlinear case (\ref{NLGen}), the solution to the Riemann problem (\ref{SystHypVect})-(\ref{SystCI}) involves two waves associated to each characteristic field (figure~\ref{fig:Notations}-(a)), which can be either a shock or a rarefaction wave~\cite{godlewski96}. In the non-convex case (\ref{NLNGen}), compound waves made of both rarefaction and discontinuity may arise~\cite{wendroff72a}. These elementary solutions{\,---\,}discontinuities, rarefactions and compound waves{\,---\,}are examined separately in the next section. For this purpose, we study $p$-waves ($p=1$ or $p=2$) which connect a left state $\bm{U}_\ell$ and a right state $\bm{U}_r$ (see figure~\ref{fig:Notations}-(b)). Analytical expressions are detailed for the models 1, 2 and 3.
	
	\begin{figure}
		\begin{minipage}{0.5\linewidth}
			\centering
			(a)
			\vspace{0.5em}
			
			\includegraphics{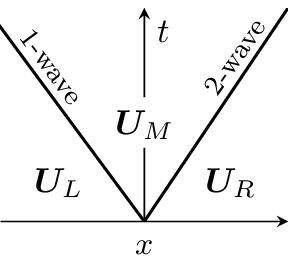}
		\end{minipage}
		\begin{minipage}{0.5\linewidth}
			\centering
			(b)
			\vspace{0.5em}
			
			\includegraphics{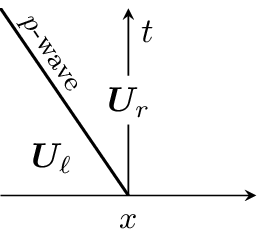}
		\end{minipage}
		\caption{(a) Structure of the solution to the Riemann problem. (b) Structure of an elementary solution in one characteristic field. If $p=1$, then $(\bm{U}_\ell, \bm{U}_r)=(\bm{U}_L, \bm{U}_M)$. If $p=2$, then $(\bm{U}_\ell, \bm{U}_r)=(\bm{U}_M, \bm{U}_R)$.\label{fig:Notations}}
	\end{figure}
		
	%%%%%%%%%%%%%%%%%%%%%%%%%%%%%%%%%%%%%%%%%%%%%%%%%%%%%%%%%%%%%%%%%%%%
	
	\section{Elementary solutions}\label{sec:ElemSol}
	
	\subsection{Discontinuities}\label{subsec:Discont}
	
	We are looking for piecewise constant solutions to the Riemann problem (\ref{SystHypVect})-(\ref{SystCI}) in one characteristic field ($p=1$ or $p=2$). They satisfy the Rankine-Hugoniot jump condition~\cite{godlewski96}
	\begin{equation}
		\bm{f}({\bm U}_r)-\bm{f}(\bm{U}_\ell) = s_p\, ({\bm U}_r - \bm{U}_\ell)\, ,
		\label{RH}
	\end{equation}
	from which one deduces
	\begin{equation}
		v_r = v_\ell - s_p\, (\varepsilon_r-\varepsilon_\ell) \, ,
		\label{EqRH1}
	\end{equation}
	with shock speeds
	\begin{equation}
		s_1 = -\sqrt{\frac{\sigma(\varepsilon_r)-\sigma(\varepsilon_\ell)}{\rho_0\, (\varepsilon_r-\varepsilon_\ell)}} \, ,
		\qquad
		s_2 = \sqrt{\frac{\sigma(\varepsilon_r)-\sigma(\varepsilon_\ell)}{\rho_0\, (\varepsilon_r-\varepsilon_\ell)}}
		\, .
		\label{EqRH2}
	\end{equation}
	As displayed on figure~\ref{fig:Lax}, the quantity $\rho_0\,{s_p}^2$ is the slope of the line connecting $(\varepsilon_\ell,\sigma(\varepsilon_\ell))$ and $(\varepsilon_r,\sigma(\varepsilon_r))$ in the $\varepsilon$-$\sigma$ plane. A discontinuity wave is the piecewise constant function defined by
	\begin{equation}
	{\bm U}(x,t)=
	\left\lbrace\!
	{\renewcommand{\arraystretch}{1.2}
		\begin{array}{ll}
		\bm{U}_\ell &\text{if } x < s_p\, t \, ,\\
		{\bm U}_r &\text{if } x > s_p\, t \, .
		\end{array}}\!\right.
	\label{SolRiemannChoc}
	\end{equation}
	It is a weak solution of the Riemann problem (\ref{SystHypVect})-(\ref{SystCI})~\cite{godlewski96}.
	
	The discontinuity (\ref{SolRiemannChoc}) may be not admissible. Indeed, such a weak solution of the Riemann problem is not necessarily the physical (entropic) solution. First, we examine the classical situation where the characteristic fields are either linearly degenerate or genuinely nonlinear.
	
	If the characteristic fields are linearly degenerate, a discontinuity is admissible if
	\begin{equation}
		\lambda_p(\bm{U}_\ell) = s_p = \lambda_p(\bm{U}_r) \, ,
		\label{LaxContact}
	\end{equation}
	i.e.
	\begin{equation}
		\sigma'(\varepsilon_\ell) = \frac{\sigma(\varepsilon_r)-\sigma(\varepsilon_\ell)}{\varepsilon_r-\varepsilon_\ell} = \sigma'(\varepsilon_r)\, .
		\label{LaxContactElasto}
	\end{equation}
	Then, the discontinuity (\ref{SolRiemannChoc}) is a contact discontinuity.
	
	If the characteristic fields are genuinely nonlinear, the discontinuity is admissible if and only if it satisfies the \emph{Lax entropy condition}~\cite{leveque02}
	\begin{equation}
		\lambda_p(\bm{U}_\ell) > s_p > \lambda_p(\bm{U}_r) \, .
		\label{Lax}
	\end{equation}
	If (\ref{Lax}) holds, then the discontinuity wave (\ref{SolRiemannChoc}) is a shock wave, and not a contact discontinuity (\ref{LaxContact}). The Lax entropy condition yields
	\begin{equation}
		\sigma'(\varepsilon_\ell) < \frac{\sigma(\varepsilon_r)-\sigma(\varepsilon_\ell)}{\varepsilon_r-\varepsilon_\ell} < \sigma'(\varepsilon_r) \qquad\text{if $p=1$,}
		\label{LaxElasto1}
	\end{equation}
	and
	\begin{equation}
		\sigma'(\varepsilon_\ell) > \frac{\sigma(\varepsilon_r)-\sigma(\varepsilon_\ell)}{\varepsilon_r-\varepsilon_\ell} > \sigma'(\varepsilon_r)
		\qquad\text{if $p=2$.}
		\label{LaxElasto2}
	\end{equation}
	An illustration is given on figure~\ref{fig:Lax} where $\sigma$ is concave. Graphically, it shows that the Lax entropy condition (\ref{Lax}) reduces to
	\begin{equation}
		s_p\,(\varepsilon_r - \varepsilon_\ell) > 0\, .
		\label{LaxElasto}
	\end{equation}
	
	\begin{figure}
		\begin{minipage}{0.5\linewidth}
			\centering
			(a)
			\vspace{0.5em}
			
			\includegraphics{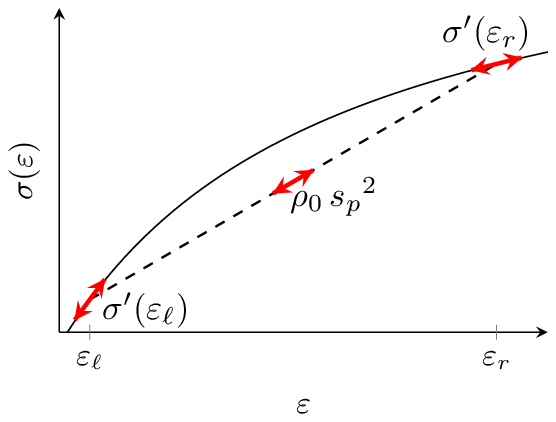}
		\end{minipage}
		\begin{minipage}{0.5\linewidth}
			\centering
			(b)
			\vspace{0.5em}
			
			\includegraphics{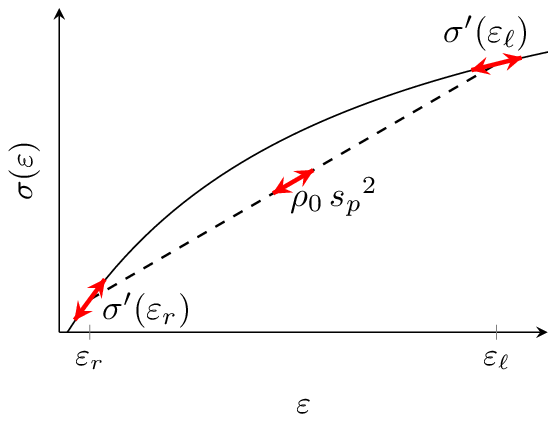}
		\end{minipage}
		\caption{Sketch of $\sigma$ between $\varepsilon_\ell$ and $\varepsilon_r$ if (a) $\varepsilon_\ell < \varepsilon_r$ and (b) $\varepsilon_\ell > \varepsilon_r$. In the case of 1-shocks, (a) is not admissible and (b) is admissible in the sense of Lax (\ref{LaxElasto1}). In the case of 2-shocks, (a) is admissible and (b) is not admissible (\ref{LaxElasto2}).\label{fig:Lax}}
	\end{figure}
	
	When the characteristic fields are neither linearly degenerate nor genuinely nonlinear, a $p$-discontinuity must satisfy the \emph{Liu entropy condition} (equation (E) in \cite{liu74}). In the case of elasticity, it writes
	\begin{equation}
	{\addtolength{\jot}{0.1em}
		\begin{aligned}
		&s_1 \geqslant -\sqrt{\frac{\sigma(\varepsilon)-\sigma(\varepsilon_\ell)}{\rho_0\, (\varepsilon-\varepsilon_\ell)}} &\text{if } p=1\, ,\\
		&s_2 \leqslant \sqrt{\frac{\sigma(\varepsilon)-\sigma(\varepsilon_\ell)}{\rho_0\, (\varepsilon-\varepsilon_\ell)}} &\text{if } p=2\, ,
		\end{aligned}}
	\label{Liu}
	\end{equation}
	for all $\varepsilon$ between $\varepsilon_\ell$ and $\varepsilon_r$. In (\ref{Liu}), $s_1$ and $s_2$ are given by (\ref{EqRH2}). In general, the Liu's entropy condition (\ref{Liu}) is stricter than Lax's shock inequalities (\ref{LaxElasto1})-(\ref{LaxElasto2}), but in the genuinely nonlinear case (\ref{NLGen}), both are equivalent. A geometrical interpretation of (\ref{Liu}) can be stated as follows (section 8.4 in \cite{dafermos05}):
	\begin{itemize}
		\item if $s_p\, (\varepsilon_r - \varepsilon_\ell)<0$, the $p$-discontinuity that joins $\bm{U}_\ell$ and $\bm{U}_r$ is admissible if the graph of $\sigma$ between $\varepsilon_\ell$ and $\varepsilon_r$ lies below the chord that connects $(\varepsilon_\ell,\sigma(\varepsilon_\ell))$ to $(\varepsilon_r,\sigma(\varepsilon_r))$;
		\item if $s_p\, (\varepsilon_r - \varepsilon_\ell)>0$, the $p$-discontinuity that joins $\bm{U}_\ell$ and $\bm{U}_r$ is admissible if the graph of $\sigma$ between $\varepsilon_\ell$ and $\varepsilon_r$ lies above the chord that connects $(\varepsilon_\ell,\sigma(\varepsilon_\ell))$ to $(\varepsilon_r,\sigma(\varepsilon_r))$.
	\end{itemize}
	
	\begin{figure}
		\centering
		\includegraphics{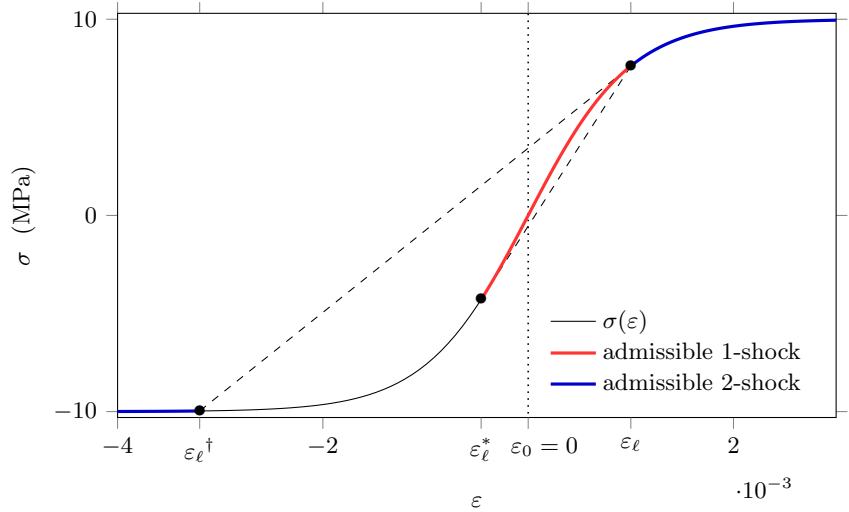}
		
		\caption{Admissibility of shocks in the sense of Liu (\ref{Liu}) for the tanh constitutive law (\ref{BLawHypTan}) with $\varepsilon_\ell=10^{-3}$.\label{fig:Liu1}}
	\end{figure}
	
	To carry out this interpretation in the nonconvex case, one needs the fonction $F$ defined for $a\neq b$ by
	\begin{equation}
		F : (a,b) \mapsto \sigma'(a) - \frac{\sigma(a)-\sigma(b)}{a-b} \, .
		\label{fonctF}
	\end{equation}
	Also, we denote by $a^\dagger$ and $b^*$ the points such that
	\begin{equation}
		F(a,a^\dagger) = 0
		\qquad\text{and}\qquad
		F(b^*,b) = 0 \, .
		\label{EqWendroff}
	\end{equation}
	By construction, one has $(a^\dagger)^* = (a^*)^\dagger = a$. Then, the geometrical interpretation of Liu's entropy condition (\ref{Liu}) is illustrated on figure~\ref{fig:Liu1}, where $\sigma$ is convex for $\varepsilon<\varepsilon_0$ and concave for $\varepsilon>\varepsilon_0$ with $\varepsilon_0=0$ (\ref{NLNGen}). On this figure, $\varepsilon_\ell$ belongs to the concave part.	
	\begin{itemize}
		\item When $p=1$, then $s_p<0$ (\ref{EqRH2}). If $\varepsilon_r>\varepsilon_\ell$, the graph must lie below the chord, which is not possible due to the concavity. If $\varepsilon_r<\varepsilon_\ell$, the graph must lie above the chord, which is only possible if $\varepsilon_r\geqslant \varepsilon_\ell^*$ (\ref{EqWendroff}), where the chord is tangent to the curve at $\varepsilon_r=\varepsilon_\ell^*$.
	
		\item When $p=2$, then $s_p>0$ (\ref{EqRH2}). If $\varepsilon_r>\varepsilon_\ell$, the graph must lie above the chord, which is satisfied due to the concavity. If $\varepsilon_r<\varepsilon_\ell$, the graph must lie below the chord, which is only possible if $\varepsilon_r\leqslant {\varepsilon_\ell}^\dagger$ (\ref{EqWendroff}), such that the chord is tangent to the curve at $\varepsilon_\ell$.
	\end{itemize}
	When $\varepsilon_\ell$ belongs to the convex part, one can carry out a similar analysis to describe the admissibility of $p$-discontinuities. The result is the same, but inequalities are of opposite sense. Finally, after multiplication by $(\varepsilon_\ell - \varepsilon_0)$, one obtains the inequalities ensuring that a $p$-discontinuity is admissible:
	\begin{equation}
	{\addtolength{\jot}{0.1em}
		\begin{aligned}
		&(\varepsilon_\ell-\varepsilon_0)\, \varepsilon_\ell^* \leqslant (\varepsilon_\ell-\varepsilon_0)\, \varepsilon_r < (\varepsilon_\ell-\varepsilon_0)\, \varepsilon_\ell
		&\text{if}\quad p=1\, ,\\
		&(\varepsilon_\ell-\varepsilon_0)\, \varepsilon_\ell < (\varepsilon_\ell-\varepsilon_0)\, \varepsilon_r
		\quad\text{or}\quad
		(\varepsilon_\ell-\varepsilon_0)\, \varepsilon_r \leqslant (\varepsilon_\ell-\varepsilon_0)\, {\varepsilon_\ell}^\dagger
		&\text{if}\quad p=2\, .
		\label{LiuConv}
		\end{aligned}}
	\end{equation}
	For more than one inflexion point, contact discontinuities (\ref{LaxContact}) may be admissible in the sense of Liu (\ref{Liu}). Here, only one inflexion point is considered. In this case, no contact discontinuity is admissible.
	
	Now, we put $\bm{U}_\ell$ in the $\varepsilon$-$v$ plane, and we construct the locus of right states $\bm{U}$ which can be connected to $\bm{U}_\ell$ through a $p$-discontinuity. The jump between $\bm{U}_\ell$ and $\bm{U}$ must satisfy the Rankine-Hugoniot condition (\ref{RH}). Thus, we obtain the curves $\mathcal{S}_p\!\left(\bm{U}_\ell\right)$ called $p$-Hugoniot loci and denoted by $\mathcal{S}_p^\ell$ for the sake of simplicity:
	\begin{equation}
		\begin{aligned}
			&v = v_\ell + \,\text{sgn}(\varepsilon-\varepsilon_\ell) \sqrt{(\sigma(\varepsilon)-\sigma(\varepsilon_\ell)) (\varepsilon-\varepsilon_\ell) / \rho_0} \equiv \mathcal{S}_1^\ell(\varepsilon) \, ,\\
			&v = v_\ell - \,\text{sgn}(\varepsilon-\varepsilon_\ell) \sqrt{(\sigma(\varepsilon)-\sigma(\varepsilon_\ell)) (\varepsilon-\varepsilon_\ell) / \rho_0} \equiv \mathcal{S}_2^\ell(\varepsilon) \, .
		\end{aligned}
		\label{EqRH}
	\end{equation}
	A few properties of these curves are detailed in appendix~\ref{subsec:WaveCurves}.
	
	\paragraph*{Model 1 {\mdseries (hyperbola)}.}
	Since $\sigma$ is concave, the Liu condition (\ref{Liu}) amounts to the Lax condition (\ref{Lax}), which reduces to $s_p\,(\varepsilon_r - \varepsilon_\ell) > 0$ (\ref{LaxElasto}).
	
	\paragraph*{Model 2 {\mdseries (tanh)}.}
	Here, $\sigma''$ is strictly decreasing and equals zero at $\varepsilon_0=0$. The stress $\sigma$ is convex if $\varepsilon<\varepsilon_0$ and concave if $\varepsilon>\varepsilon_0$. Therefore, Liu's entropy condition reduces to (\ref{LiuConv}). An illustration is given on figure~\ref{fig:Liu1}, where $\varepsilon_\ell = 10^{-3}$.
	
	\paragraph*{Model 3 {\mdseries (Landau)}.}
	If $\delta=0$ in (\ref{BLawLandau}), $\sigma$ is concave. Similarly to model~1, a $p$-shock is admissible if $s_p\,(\varepsilon_r - \varepsilon_\ell) > 0$ (\ref{LaxElasto}). Else, $\sigma''$ is strictly decreasing and equals zero at $\varepsilon_0 = -\beta/3\delta$. Then, Landau's law is similar to the tanh (model~2), and Liu's condition implies (\ref{LiuConv}).
	
	\subsection{Rarefaction waves}\label{subsec:Rar}
	
	We are looking for piecewise smooth continuous solutions of (\ref{SystHypVect})-(\ref{SystCI}) which connect $\bm{U}_\ell$ and $\bm{U}_r$. Since the system of conservation laws is invariant under uniform stretching of space and time coordinates $(x,t) \mapsto (\alpha x,\alpha t)$, we restrict ourselves to self-similar solutions of the form
	\begin{equation}
		\bm{U}(x,t)=\bm{V}(\xi) , \quad\text{where}\quad \xi=x/t\, .
		\label{SimpleAutosimilar}
	\end{equation}
	Injecting (\ref{SimpleAutosimilar}) in (\ref{SystHypVect}) gives two equations satisfied by $\bm{V}'(\xi)$. The trivial solution $\bm{V}'(\xi)=\bm{0}$ is eliminated. Differentiating the other equation implies that there exists $p\in\left\lbrace 1, 2\right\rbrace$, such that (section I.3.1 in \cite{godlewski96})
	\begin{equation}
		\left\lbrace
		\begin{aligned}
		&\lambda_p\!\left(\bm{V}(\xi)\right)=\xi \, ,\\
		&\bm{V}'(\xi) = \frac{1}{\bm{\nabla}\lambda_p\!\left(\bm{V}(\xi)\right) \!\cdot  \bm{r}_p\!\left(\bm{V}(\xi)\right)}\,\bm{r}_p\!\left(\bm{V}(\xi)\right) \, .
		\end{aligned}\right.
		\label{EqRarefaction1}
		\end{equation}
	To connect left and right states, we impose that $\bm{V}(\lambda_p(\bm{U}_\ell))={\bm U}_\ell$ and $\bm{V}(\lambda_p(\bm{U}_r))=\bm{U}_r$. Then, the function
	\begin{equation}
		\bm{U}(x,t)=
		\left\lbrace\!
		{\renewcommand{\arraystretch}{1.2}
			\begin{array}{ll}
			\bm{U}_\ell &\text{if } x \leqslant \lambda_p(\bm{U}_\ell)\, t \, ,\\
			\bm{V}(x/t) &\text{if } \lambda_p(\bm{U}_\ell)\, t \leqslant x \leqslant \lambda_p(\bm{U}_r)\, t \, ,\\
			{\bm U}_r &\text{if } \lambda_p(\bm{U}_r)\, t \leqslant x \, ,
			\end{array}}\!\right.
		\label{SolRiemannRarefaction}
	\end{equation}
	is a self-similar weak solution of (\ref{SystHypVect})-(\ref{SystCI}) connecting ${\bm U}_\ell$ and $\bm{U}_r$~\cite{godlewski96}. Such a solution is called simple wave or rarefaction wave. To be admissible, the eigenvalue $\lambda_p(\bm{V}(\xi))$ must be increasing from $\xi = \lambda_p(\bm{U}_\ell)$ to $\xi = \lambda_p(\bm{U}_r)$. In particular, we must have
	\begin{equation}
		\lambda_p(\bm{U}_\ell) \leqslant \lambda_p(\bm{U}_r) \, .
		\label{pseudoLax}
	\end{equation}
	Furthermore, equation (\ref{EqRarefaction1}) requires that $\bm{\nabla} \lambda_p\cdot \bm{r}_p$ does not vanish along the curve $\xi\mapsto \bm{V}(\xi)$. This is never satisfied when the characteristic fields are linearly degenerate, but it is always satisfied when the characteristic fields are genuinely nonlinear. In the nonconvex case (\ref{NLNGen}), it implies that a rarefaction cannot cross the inflection point $\varepsilon_0$:
	\begin{equation}
		(\varepsilon_\ell-\varepsilon_0)(\varepsilon_r-\varepsilon_0)\geqslant 0 \, .
		\label{pseudoLaxNGNL}
	\end{equation}
	
	Let us define a primitive $C$ of the sound speed $c$ over $\left]\varepsilon_\textit{inf},\varepsilon_\text{\it sup}\right[$. Then, the $p$-Riemann invariants
	\begin{equation}
		w_1\!\left(\bm{U}\right) = v - C(\varepsilon)\, ,\qquad
		w_2\!\left(\bm{U}\right) = v + C(\varepsilon)\, ,
		\label{RiemannInvar}
	\end{equation}
	are constant on $p$-rarefaction waves \cite{godlewski96}. In practice, this property is used to rewrite (\ref{EqRarefaction1}) as
	\begin{equation}
	\left\lbrace
	{\addtolength{\jot}{0.1em}
		\begin{aligned}
		&\lambda_p\!\left(\bm{V}(\xi)\right)=\xi\, ,\\
		&w_p(\bm{V}(\xi)) = w_p(\bm{U}_\ell)\, .
		\end{aligned}}\right.
	\label{EqRarefaction2}
	\end{equation}
	Finally, using the expressions of the eigenvalues (\ref{SystHypValP}) and the Riemann invariants (\ref{RiemannInvar}), one obtains
	\begin{equation}
	{\addtolength{\jot}{0.1em}
		\begin{aligned}
		&\bm{V}(\xi) =
		%		\left(
		%		{\renewcommand{\arraystretch}{1.2}
		%		\!\begin{array}{c}
		%			\varepsilon(\xi)\\
		%			v(\xi)
		%		\end{array}\!}\right) =
		\left(
		{\renewcommand{\arraystretch}{1.2}
			\!\begin{array}{c}
			c^{-1} \!\left(-\xi\right)\\
			w_1(\bm{U}_\ell) + C\circ c^{-1} \!\left(-\xi\right)
			\end{array}\!}\right) &\text{if } p=1,
		\\
		&\bm{V}(\xi) =
		%		\left(
		%		{\renewcommand{\arraystretch}{1.2}
		%		\!\begin{array}{c}
		%			\varepsilon(\xi)\\
		%			v(\xi)
		%		\end{array}\!}\right) =
		\left(
		{\renewcommand{\arraystretch}{1.2}
			\!\begin{array}{c}
			c^{-1} \!\left(\xi\right)\\
			w_2(\bm{U}_\ell) - C\circ c^{-1} \!\left(\xi\right)
			\end{array}\!}\right) &\text{if } p=2.
		\end{aligned}}
	\label{SolRarefaction}
	\end{equation}
	In (\ref{EqRarefaction2})-(\ref{SolRarefaction}), $\bm{U}_\ell$ can be replaced by $\bm{U}_r$, or by any other state on the rarefaction wave.
	
	Now, we put $\bm{U}_\ell$ in the $\varepsilon$-$v$ plane, and we construct the locus of right states $\bm{U}$ which can be connected to $\bm{U}_\ell$ through a $p$-rarefaction. The states $\bm{U}_\ell$ and $\bm{U}$ must satisfy $w_p\!\left(\bm{U}\right) = w_p\!\left(\bm{U}_\ell\right)$. Thus, we obtain the rarefaction curves $\mathcal{R}_p\!\left(\bm{U}_\ell\right)$ and denoted by $\mathcal{R}_p^\ell$ for the sake of simplicity:
	\begin{equation}
		{\addtolength{\jot}{0.1em}
		\begin{aligned}
			& v = v_\ell - C(\varepsilon_\ell) + C(\varepsilon) \equiv \mathcal{R}_1^\ell(\varepsilon)\, ,\\
			& v = v_\ell + C(\varepsilon_\ell) - C(\varepsilon) \equiv \mathcal{R}_2^\ell(\varepsilon)\, .
		\end{aligned}}
		\label{EqRarefaction}
	\end{equation}
	A few properties of these curves are detailed in appendix~\ref{subsec:WaveCurves}.
	
	\paragraph*{Model 1 {\mdseries (hyperbola)}.}
	To compute rarefaction waves, one needs the expressions of $C$ and $c^{-1}$ in (\ref{SolRarefaction}). For the hyperbola law, a primitive of the sound speed (\ref{CHyp}) is
	\begin{equation}
		C(\varepsilon)=d\, c_0 \ln (1+\varepsilon/d)\, ,
		\label{CintHyp}
	\end{equation}
	and the inverse function of $c$ is
	\begin{equation}
		c^{-1}(\xi)=d\left(\frac{c_0}{\xi}-1\right) .
		\label{CinvHyp}
	\end{equation}
	
	\paragraph*{Model 2 {\mdseries (tanh)}.}
	A primitive of the sound speed (\ref{CHypTan}) is
	\begin{equation}
		C(\varepsilon)= c_0\,d\,\arcsin(\tanh(\varepsilon/d))\, .
		\label{CintHypTan}
	\end{equation}
	Since $c$ is not monotonous (figure~\ref{fig:HypTan}-(b)), its inverse is not unique. The inverse over the range $[0,c_0]$ is made of two branches:
	\begin{equation}
		c^{-1}(\xi) \in \left\lbrace -d\,\text{arcosh}\!\left(\frac{c_0}{\xi}\right)\! , d\,\text{arcosh}\!\left(\frac{c_0}{\xi}\right) \right\rbrace .
		\label{CinvHypTan}
	\end{equation}
	The choice of the inverse (\ref{CinvHypTan}) in (\ref{SolRarefaction}) depends on $\varepsilon_\ell$. Indeed, $\bm{V}(\xi)$ must satisfy $\bm{V}(\lambda_p(\bm{U}_\ell))={\bm U}_\ell$ and $\bm{V}(\lambda_p(\bm{U}_r))=\bm{U}_r$, i.e. $\varepsilon_\ell = c^{-1}\circ c(\varepsilon_\ell)$ and $\varepsilon_r = c^{-1}\circ c(\varepsilon_r)$. Since $\varepsilon_\ell$ and $\varepsilon_r$ are on the same side of the inflection point (\ref{pseudoLaxNGNL}), the choice of the inverse in (\ref{CinvHypTan}) relies only on $\varepsilon_\ell$. If $\varepsilon_\ell<\varepsilon_0$, the inverse (\ref{CinvHypTan}) must be lower than $\varepsilon_0=0$ (first expression). Else, it must be larger (second expression).
	
	\paragraph*{Model 3 {\mdseries (Landau)}.}
	In the case of the quadratic nonlinearity ($\delta=0$), a primitive of the sound speed (\ref{CLandau}) is
	\begin{equation}
		C(\varepsilon)={-c_0}\, \frac{(1-2\beta\,\varepsilon)^{3/2}}{3\beta}\, ,
		\label{CintLandau1}
	\end{equation}
	and the inverse function of $c$ is
	\begin{equation}
		c^{-1}(\xi)=
		\frac{{c_0}^2-\xi^2}{2\beta\, {c_0}^2}\, .
		\label{CinvLandau1}
	\end{equation}
	In the case of the cubic nonlinearity ($\delta\neq 0$), a primitive of the sound speed (\ref{CLandau}) is
	\begin{equation}
		C(\varepsilon)=c(\varepsilon)\, \frac{\beta+3\delta\,\varepsilon}{6\delta} + c_0\, \frac{\beta^2+3\delta}{6\delta\sqrt{3\delta}} \arcsin\!\left(\frac{\beta+3\delta\,\varepsilon}{\sqrt{\beta^2+3\delta}}\right) .
		\label{CintLandau}
	\end{equation}
	Here too, $c$ is not monotonous (figure~\ref{fig:Landau}-(d)). The inverse over the range $[0,c(\varepsilon_0)]$ (see (\ref{CLandauMax})) is made of two branches:
	\begin{equation}
		c^{-1}(\xi) \in \left\lbrace
		-\frac{\beta}{3\delta} - \sqrt{\frac{\beta^2}{9\delta^2} + \frac{1}{3\delta} \!\left(1 - \frac{\xi^2}{{c_0}^2}\right)\!  } \, ,
		-\frac{\beta}{3\delta} + \sqrt{\frac{\beta^2}{9\delta^2} + \frac{1}{3\delta} \!\left(1 - \frac{\xi^2}{{c_0}^2}\right)\!  }\, \right\rbrace .
		\label{CinvLandau}
	\end{equation}
	The choice of the inverse in (\ref{SolRarefaction}) depends on $\varepsilon_\ell$. If $\varepsilon_\ell<\varepsilon_0$, the inverse (\ref{CinvLandau}) must be lower than $\varepsilon_0=-\beta/3\delta$ (first expression). Else, it must be larger (second expression).
	
	\subsection{Compound waves}\label{subsec:Compound}
	
	In this section, $\sigma$ has an inflection point at $\varepsilon_0$ (\ref{NLNGen}). The characteristic fields are thus not genuinely nonlinear over $\left]\varepsilon_\textit{inf}, \varepsilon_\text{\it sup}\right[$. On the one hand, a $p$-discontinuity which crosses the line $\varepsilon=\varepsilon_0$ is not always admissible (\ref{LiuConv}). On the other hand, a $p$-rarefaction cannot cross the line $\varepsilon=\varepsilon_0$ (\ref{pseudoLaxNGNL}). When discontinuities and rarefactions are not admissible, one can start from $\bm{U}_\ell$ with an admissible $p$-discontinuity and connect it to $\bm{U}_r$ with an admissible $p$-rarefaction (shock-rarefaction). Alternatively, one can start from $\bm{U}_\ell$ with an admissible $p$-rarefaction and connect it to $\bm{U}_r$ with an admissible $p$-discontinuity (rarefaction-shock). These compound waves composed of one rarefaction and one discontinuity are now examined separately.
	
	\paragraph*{Shock-rarefactions.}
	We consider a $p$-shock-rarefaction that connects $\bm{U}_\ell$ and $\bm{U}_r$. The rarefaction cannot cross the line $\varepsilon=\varepsilon_0$. It breaks when reaching $\varepsilon_\ell^*$ \cite{wendroff72a} such that $F(\varepsilon_\ell^*,\varepsilon_\ell)=0$ (\ref{EqWendroff}). Therefore, a shock-rarefaction is defined by
	\begin{equation}
	\bm{U}(x,t)=
	\left\lbrace\!
	{\renewcommand{\arraystretch}{1.2}
		\!\begin{array}{ll}
		\bm{U}_\ell &\text{if } x < \lambda_p(\varepsilon_\ell^*)\, t\, , \\
		{\bm V}(x/t) &\text{if } \lambda_p(\varepsilon_\ell^*)\, t < x \leqslant \lambda_p(\varepsilon_r)\, t\, , \\
		{\bm U}_r &\text{if } \lambda_p(\varepsilon_r)\, t \leqslant x\, .
		\end{array}}\!\right.
	\label{SolRiemannChocR}
	\end{equation}
	$\bm{V}(\xi)$ is given by (\ref{SolRarefaction}) where $\bm{U}_\ell$ has to be replaced by $\bm{U}_r$. An illustration is given on figure~\ref{fig:Compound}-(a), where the parameters are the same as in figure~\ref{fig:Lan1} (section~\ref{sec:NumEx}). If the shock-rarefaction (\ref{SolRiemannChocR}) is a weak solution of (\ref{SystHypVect})-(\ref{SystCI}), then both parts are weak solutions. On the one hand, the discontinuous part must satisfy the Rankine-Hugoniot condition (\ref{RH}) with left state $\bm{U}_\ell$ and right state $\bm{U}_\ell^* = (\varepsilon_\ell^*,v_\ell^*)^\top$:
	\begin{equation}
		v_\ell^* = v_\ell - s_p\,(\varepsilon_\ell^* - \varepsilon_\ell)\, .
		\label{SR-RH}
	\end{equation}
	Due to the relation (\ref{EqWendroff}) between $\varepsilon_\ell$ and $\varepsilon_\ell^*$, the shock speed $s_p$ (\ref{EqRH2}) satisfies
	\begin{equation}
		s_1 = -\sqrt{\frac{\sigma(\varepsilon_\ell^*)-\sigma(\varepsilon_\ell)}{\rho_0\, (\varepsilon_\ell^*-\varepsilon_\ell)}} = -c(\varepsilon_\ell^*) \, ,\qquad s_2 = \sqrt{\frac{\sigma(\varepsilon_\ell^*)-\sigma(\varepsilon_\ell)}{\rho_0\, (\varepsilon_\ell^*-\varepsilon_\ell)}} = c(\varepsilon_\ell^*) \, .
	\end{equation}
	On the other hand, the Riemann invariants (\ref{RiemannInvar}) must be constant on the continuous part:
	\begin{equation}
		w_p(\bm{U}_\ell^*) = w_p(\bm{U}_r)\, .
		\label{SR-Rar}
	\end{equation}
	Finally, equations (\ref{SR-RH}) and (\ref{SR-Rar}) yield
	\begin{equation}
		v_r = v_\ell - ({-1})^p \left(C(\varepsilon_r) - C(\varepsilon_\ell^*) + c(\varepsilon_\ell^*)(\varepsilon_\ell^*-\varepsilon_\ell)\right) .
		\label{pseudoRH_SR}
	\end{equation}
	Admissibility of shock-rarefactions is presented in section~\ref{subsec:GraphMeth}.
	
	Now, we put $\bm{U}_\ell$ in the $\varepsilon$-$v$ plane, and we construct the locus of right states $\bm{U}$ which can be connected to $\bm{U}_\ell$ through a $p$-shock-rarefaction. The states $\bm{U}_\ell$ and $\bm{U}$ must satisfy (\ref{pseudoRH_SR}). Thus, we obtain the shock-rarefaction curves $\mathcal{SR}_p\!\left(\bm{U}_\ell\right)$ and denoted by $\mathcal{SR}_p^\ell$ for the sake of simplicity:
	\begin{equation}
		\begin{aligned}
			&v = v_\ell + c(\varepsilon_\ell^*) (\varepsilon_\ell^* - \varepsilon_\ell) - C(\varepsilon_\ell^*) + C(\varepsilon) \equiv \mathcal{SR}_1^\ell(\varepsilon) \, ,\\
			&v = v_\ell - c(\varepsilon_\ell^*) (\varepsilon_\ell^* - \varepsilon_\ell) + C(\varepsilon_\ell^*) - C(\varepsilon) \equiv \mathcal{SR}_2^\ell(\varepsilon) \, .
		\end{aligned}
		\label{EqSR}
	\end{equation}
	
	\paragraph*{Rarefaction-shocks.}
	We consider a $p$-rarefaction-shock that connects $\bm{U}_\ell$ and $\bm{U}_r$. The rupture of the rarefaction wave occurs when reaching $\varepsilon_r^*$ \cite{wendroff72a} such that $F(\varepsilon_r^*,\varepsilon_r)=0$ (\ref{EqWendroff}). Therefore, a rarefaction-shock is defined by
	\begin{equation}
		\bm{U}(x,t)=
		\left\lbrace\!
		{\renewcommand{\arraystretch}{1.2}
			\!\begin{array}{ll}
			\bm{U}_\ell &\text{if } x \leqslant \lambda_p(\varepsilon_\ell)\, t\, ,\\
			\bm{V}(x/t) &\text{if } \lambda_p(\varepsilon_\ell)\, t\leqslant x < \lambda_p(\varepsilon_r^*)\, t\, ,\\
			{\bm U}_r &\text{if } \lambda_p(\varepsilon_r^*)\, t < x\, ,
			\end{array}}\!\right.
		\label{SolRiemannRChoc}
	\end{equation}
	where $\bm{V}(\xi)$ is given by (\ref{SolRarefaction}). An illustration is given on figure~\ref{fig:Compound}-(b), where the parameters are the same as in figure~\ref{fig:Lan1}. With similar arguments than for (\ref{SR-RH}) and (\ref{SR-Rar}), one obtains
	\begin{equation}
		v_r = v_\ell + ({-1})^p \left(C(\varepsilon_\ell) - C(\varepsilon_r^*) + c(\varepsilon_r^*)(\varepsilon_r^*-\varepsilon_r)\right) .
		\label{pseudoRH_RS}
	\end{equation}
	Admissibility of rarefaction-shocks is presented in section~\ref{subsec:GraphMeth}, where the computation of $\varepsilon^*$ is also examined.
	
	\begin{figure}
		\begin{minipage}{0.52\linewidth}
			\centering
			(a)
			
			%\flushright
			\includegraphics{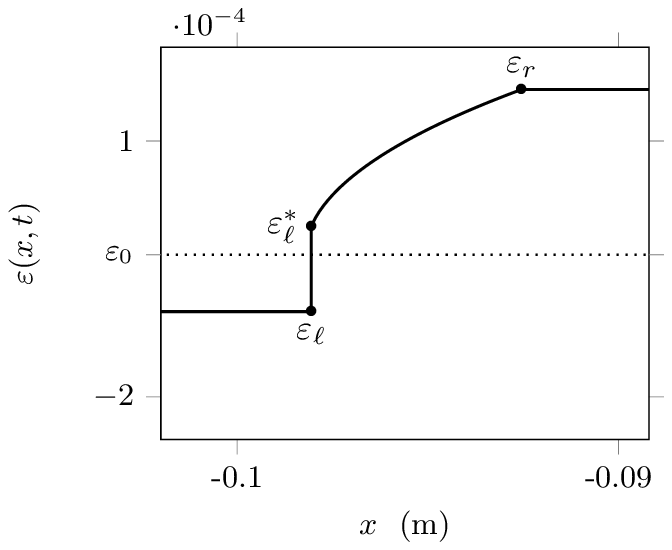}
		\end{minipage}
		\begin{minipage}{0.47\linewidth}
			\centering
			(b)
			
			%\flushleft
			\includegraphics{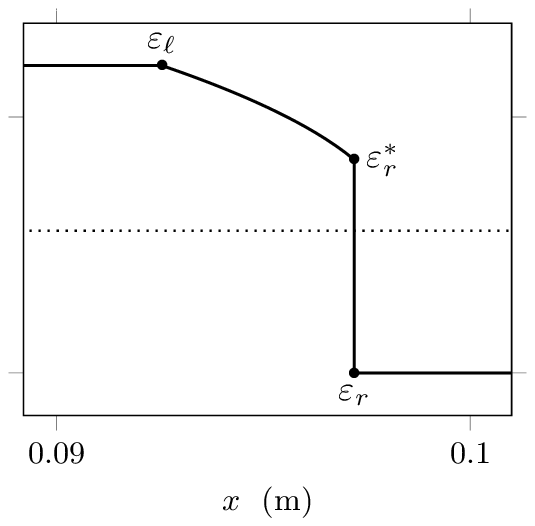}
		\end{minipage}
		
		\caption{Compound waves obtained with the cubic Landau's law (\ref{BLawLandau}). Snapshot of $\varepsilon$ in the case of (a) a 1-shock-rarefaction and (b) a 2-rarefaction-shock.\label{fig:Compound}}
	\end{figure}
	
	Now, we put $\bm{U}_\ell$ in the $\varepsilon$-$v$ plane, and we construct the locus of right states $\bm{U}$ which can be connected to $\bm{U}_\ell$ through a $p$-rarefaction-shock. The states $\bm{U}_\ell$ and $\bm{U}$ must satisfy (\ref{pseudoRH_RS}). Thus, we obtain the rarefaction-shock curves $\mathcal{RS}_p\!\left(\bm{U}_\ell\right)$ and denoted by $\mathcal{RS}_p^\ell$ for the sake of simplicity:
	\begin{equation}
		\begin{aligned}
			& v = v_\ell - C(\varepsilon_\ell) + C(\varepsilon^*) - c(\varepsilon^*) (\varepsilon^* - \varepsilon) \equiv \mathcal{RS}_1^\ell(\varepsilon) \, ,\\
			& v = v_\ell + C(\varepsilon_\ell) - C(\varepsilon^*) + c(\varepsilon^*) (\varepsilon^* - \varepsilon) \equiv \mathcal{RS}_2^\ell(\varepsilon) \, .
			\end{aligned}
		\label{EqRS}
	\end{equation}
	A few properties of these curves are detailed in appendix~\ref{subsec:WaveCurves}.
	
	\subsection{Graphical method}\label{subsec:GraphMeth}
	
	In practice, a graphical method can be applied to construct entropic elementary solutions to (\ref{SystHypVect})-(\ref{SystCI}) based on discontinuities, rarefactions and compound waves. This method is very useful for nonconvex constitutive equations $\varepsilon\mapsto \sigma(\varepsilon)$ and can be stated as follows (section 9.5 in \cite{dafermos05}):
	
	\vspace{1em}
	\noindent For 1-waves,
	\begin{itemize}
		\item if $\varepsilon_r < \varepsilon_\ell$, we construct the convex hull of $\sigma$ over $\left[\varepsilon_r, \varepsilon_\ell\right]$;
		\item if $\varepsilon_r > \varepsilon_\ell$, we construct the concave hull of $\sigma$ over $\left[\varepsilon_\ell, \varepsilon_r\right]$.
	\end{itemize}
	For 2-waves,
	\begin{itemize}
		\item if $\varepsilon_r < \varepsilon_\ell$, we construct the concave hull of $\sigma$ over $\left[\varepsilon_r, \varepsilon_\ell\right]$;
		\item if $\varepsilon_r > \varepsilon_\ell$, we construct the convex hull of $\sigma$ over $\left[\varepsilon_\ell, \varepsilon_r\right]$.
	\end{itemize}
	Between $\varepsilon_\ell$ and $\varepsilon_r$, the intervals where the slope of the hull is constant correspond to admissible discontinuities. The other intervals correspond to admissible rarefactions.
	
	On figure~\ref{fig:Liu2}, we illustrate the method for the tanh constitutive law (\ref{BLawHypTan}), where the inflexion point (\ref{NLNGen}) is $\varepsilon_0=0$. $\sigma$ is convex for $\varepsilon< \varepsilon_0$ and concave for $\varepsilon> \varepsilon_0$. Here, $\varepsilon_r={-1.7}\times 10^{-3}$ is smaller than $\varepsilon_\ell=1.2\times 10^{-3}$. If $p=1$, we construct the convex hull of $\sigma$, i.e. the biggest convex fonction which is lower or equal to $\sigma$. If $p=2$, we construct the concave hull of $\sigma$, i.e. the smallest concave fonction which is greater or equal to $\sigma$. The method predicts that a compound wave can either be a 1-shock-rarefaction or a 2-rarefaction-shock. Also, it is in agreement with the definitions of shock-rarefactions and rarefaction-shocks in the previous section, since the 1-rarefaction breaks when reaching $\varepsilon_\ell^*$ and the 2-rarefaction breaks when reaching $\varepsilon_r^*$.
	
	\begin{figure}
		\centering
		\includegraphics{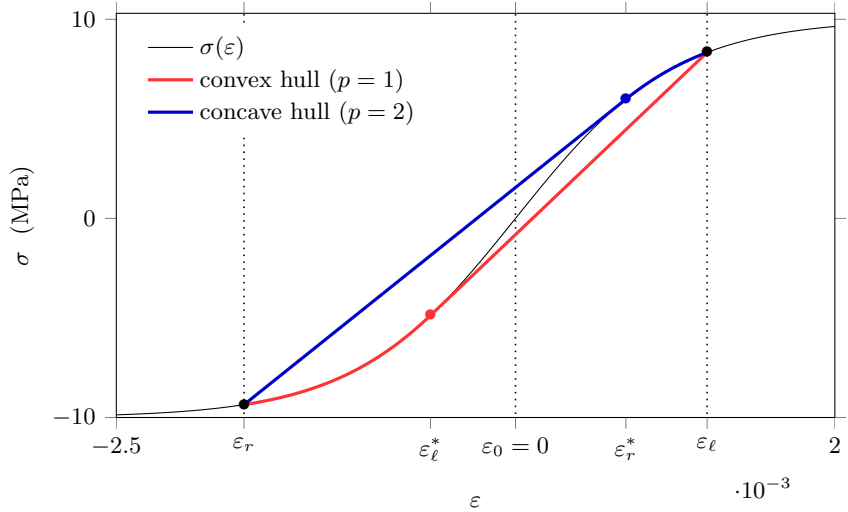}
		
		\caption{Construction of the solution for $\varepsilon_r<\varepsilon_\ell$ and the tanh constitutive law (\ref{BLawHypTan}). Here, we obtain a 1-shock-rarefaction (convex hull) and a 2-rarefaction-shock (concave hull).\label{fig:Liu2}}
	\end{figure}
		
	When $\varepsilon_r$ varies, the hulls on figure~\ref{fig:Liu2} vary. Depending on $\varepsilon_r$, one obtains different admissible waves (cf. table~\ref{tab:AdmWaves}).
	
	\begin{table}
		\caption{Admissible elementary waves for the tanh law (\ref{BLawHypTan}) when $\varepsilon_\ell>\varepsilon_0$ and $\varepsilon_r$ varies in $\mathbb{R}$ (increasing values of $\varepsilon_r$ from the left to the right).\label{tab:AdmWaves}}
		\centering
		{\begin{tabular}{lccccccccccc}
				\toprule
				& ${-\infty}$ & & ${\varepsilon_\ell}^\dagger$ & & $\varepsilon_\ell^*$ & & $\varepsilon_0=0$ & & $\varepsilon_\ell$ & & ${+\infty}$ \\
				\midrule
				$p=1$ & & $\mathcal{SR}_1$ & & $\mathcal{SR}_1$ & & $\mathcal{S}_1$ & & $\mathcal{S}_1$ & & $\mathcal{R}_1$ & \\
				$p=2$ & & $\mathcal{S}_2$ & & $\mathcal{RS}_2$ & & $\mathcal{RS}_2$ & & $\mathcal{R}_2$ & & $\mathcal{S}_2$ & \\
				\bottomrule
			\end{tabular}}
		\end{table}
		
	\paragraph*{Model 2 {\mdseries (tanh)}.}
	To compute the solution (\ref{SolRiemannRChoc}) or (\ref{SolRiemannChocR}), we need to solve (\ref{EqWendroff}). For the tanh constitutive law (\ref{BLawHypTan}), it yields
	\begin{equation}
		\sinh\!\left(\frac{\varepsilon^*-\varepsilon}{d}\right)\, - \frac{\cosh(\varepsilon/d)}{\cosh(\varepsilon^*\! /d)}\, \frac{\varepsilon^*-\varepsilon}{d} = 0 \, ,
		\label{EqWendroffHypTan}
	\end{equation}
	which can be solved iteratively, e.g. with Newton's method and the initial value $\varepsilon^*\simeq -\varepsilon/2$. This initial guess has been deduced from a Taylor expansion of (\ref{EqWendroffHypTan}) at $\varepsilon^*=\varepsilon=0$. 
		
	\paragraph*{Model 3 {\mdseries (Landau)}.}
	When $\delta=0$ in (\ref{BLawLandau}), $\sigma$ is concave and there are no compound waves. When $\delta > 0$ in (\ref{BLawLandau}), we are in a similar configuration than with tanh. Here, (\ref{EqWendroff}) can be solved analytically:
	\begin{equation}
		\varepsilon^* = -\frac{\varepsilon+\beta/\delta}{2}\, .
		\label{EqWendroffLandau}
	\end{equation}
	
	%%%%%%%%%%%%%%%%%%%%%%%%%%%%%%%%%%%%%%%%%%%%%%%%%%%%%%%%%%%%%%%%%%%%
	%-------------------------- riemann --------------------------------
	
	\section{Solution of the Riemann problem}\label{sec:Solution}
	
	\subsection{General strategy}\label{subsec:GenStrat}
	
	When (\ref{SystHypVect}) is a strictly hyperbolic system, the solution to the Riemann problem (\ref{SystHypVect})-(\ref{SystCI}) has three constant states $\bm{U}_L$, $\bm{U}_M$ and $\bm{U}_R$ (see figure~\ref{fig:Notations}-(a)). Here, every possible wave structure combining a 1-wave and a 2-wave must be examined. Since $\sigma$ has only one inflection point (\ref{NLNGen}), compound waves can only be composed of one rarefaction and one discontinuity.
	
	In order to find the intermediate state $\bm{U}_M$, we construct the forward wave curve $\Phi_p^L$ of right states $\bm{U}$ which can be connected to $\bm{U}_L$ through an admissible $p$-wave (sections 9.4-9.5 in \cite{dafermos05}). It satisfies:
	\begin{equation}
	\Phi_p^L(\varepsilon) =
	\left\lbrace
	{\addtolength{\jot}{0.1em}
		\begin{aligned}
		&\mathcal{S}_p^L(\varepsilon) & &\text{if admissible $p$-shock,} \\
		&\mathcal{R}_p^L(\varepsilon) & &\text{if admissible $p$-rarefaction,} \\
		&\mathcal{RS}_p^L(\varepsilon) & &\text{if admissible $p$-rarefaction-shock,} \\
		&\mathcal{SR}_p^L(\varepsilon) & &\text{if admissible $p$-shock-rarefaction.}
		\end{aligned}}
	\right.
	\label{Phip}
	\end{equation}
	According to equations (\ref{EqRH}), (\ref{EqRarefaction}), (\ref{EqRS}) and (\ref{EqSR}), this curve is only translated vertically when $v_L$ changes. Similarly, we construct the backward wave curve $\Psi_p^R$ of left states $\bm{U}$ which can be connected to $\bm{U}_R$ through an admissible $p$-wave:
	\begin{equation}
	\Psi_p^R(\varepsilon) =
	\left\lbrace
	{\addtolength{\jot}{0.1em}
		\begin{aligned}
		&\mathcal{S}_p^R(\varepsilon) & &\text{if admissible $p$-shock,} \\
		&\mathcal{R}_p^R(\varepsilon) & &\text{if admissible $p$-rarefaction,} \\
		&\mathcal{SR}_p^R(\varepsilon) & &\text{if admissible $p$-rarefaction-shock,} \\
		&\mathcal{RS}_p^R(\varepsilon) & &\text{if admissible $p$-shock-rarefaction.}
		\end{aligned}}
	\right.
	\label{Psip}
	\end{equation}
	Backward wave curves (\ref{Psip}) are obtained by replacing the elementary forward wave curves in (\ref{Phip}) by elementary backward wave curves. It amounts to replace rarefaction-shock curves by shock-rarefaction curves, and vice versa. Here too, the curve $\Psi_p^R$ is only translated vertically when $v_R$ changes. Also, one can remark that $v_R = \Phi_p^L (\varepsilon_R)$ is equivalent to $v_L = \Psi_p^R(\varepsilon_L)$.
	
	The intermediate state $\bm{U}_M$ is connected to $\bm{U}_L$ through an admissible 1-wave and to $\bm{U}_R$ through an admissible 2-wave. Thus, it satisfies
	\begin{equation}
		v_M = \Phi_1^L(\varepsilon_M) = \Psi_2^R(\varepsilon_M) \, ,
		\label{IntermediateState}
	\end{equation}
	or equivalently,
	\begin{equation}
	\left\lbrace
	{\addtolength{\jot}{0em}
		\begin{aligned}
		&v_M = \Phi_1^L(\varepsilon_M) \, ,\\
		&v_R = \Phi_2^M(\varepsilon_R) \, .
		\end{aligned}}\right.
	\label{IntermediateState2}
	\end{equation}
	The existence of the solution to (\ref{IntermediateState}) will be discussed in the next sections. If the solution exists, one can find the intermediate state $\bm{U}_M$ numerically. To do so, $\varepsilon_M$ is computed by solving (\ref{IntermediateState}) with the Newton-Raphson method, and by computing $v_M=\Phi_1^L(\varepsilon_M)$. The form of the solution $\bm{U}(x,t)$ is then deduced from the corresponding elementary solutions (\ref{SolRiemannChoc}), (\ref{SolRiemannRarefaction}), (\ref{SolRiemannChocR}) or (\ref{SolRiemannRChoc}).
	
	\subsection{Concave constitutive laws}\label{subsec:Concave}
	
	Let us assume that $\sigma''$ is strictly negative over $\left]\varepsilon_\text{\it inf}, \varepsilon_\text{\it sup}\right[$. Therefore, the characteristic fields are genuinely nonlinear and $\sigma$ is strictly concave. In this case, compound waves are not admissible. Also, discontinuities and rarefactions have to satisfy the admissibility conditions (\ref{LaxElasto}) and (\ref{pseudoLax}) respectively. Thus, forward and backward wave curves become
	\begin{equation}
	{\addtolength{\jot}{0.1em}
		\begin{aligned}
		&\Phi_1^L(\varepsilon) =
		\left\lbrace
		{%\addtolength{\jot}{0.1em}
			\begin{aligned}
			&\mathcal{S}_1^L(\varepsilon) & &\text{if } \varepsilon < \varepsilon_L \, ,\\
			&\mathcal{R}_1^L(\varepsilon) & &\text{if } \varepsilon \geqslant \varepsilon_L \, ,
			\end{aligned}}
		\right.
		& &\Phi_2^L(\varepsilon) =
		\left\lbrace
		{%\addtolength{\jot}{0.1em}
			\begin{aligned}
			&\mathcal{S}_2^L(\varepsilon) & &\text{if } \varepsilon > \varepsilon_L \, ,\\
			&\mathcal{R}_2^L(\varepsilon) & &\text{if } \varepsilon \leqslant \varepsilon_L \, ,
			\end{aligned}}
		\right.\\
		&\Psi_1^R(\varepsilon) =
		\left\lbrace
		{%\addtolength{\jot}{0.1em}
			\begin{aligned}
			&\mathcal{S}_1^R(\varepsilon) & &\text{if } \varepsilon > \varepsilon_R \, ,\\
			&\mathcal{R}_1^R(\varepsilon) & &\text{if } \varepsilon \leqslant \varepsilon_R \, ,
			\end{aligned}}
		\right.
		& &\Psi_2^R(\varepsilon) =
		\left\lbrace
		{%\addtolength{\jot}{0.1em}
			\begin{aligned}
			&\mathcal{S}_2^R(\varepsilon) & &\text{if } \varepsilon < \varepsilon_R \, ,\\
			&\mathcal{R}_2^R(\varepsilon) & &\text{if } \varepsilon \geqslant \varepsilon_R \, .
			\end{aligned}}
		\right.
		\end{aligned}}
	\label{WNLGen}
	\end{equation}
	
	Since the characteristic fields are genuinely nonlinear, $\Phi_1^L$ and $\Psi_2^R$ are of class $C^2$ (section I.6 in~\cite{godlewski96}). From the properties of each elementary curve studied before, we deduce that $\Phi_1^L$ is an increasing bijection over $\left]\varepsilon_\text{\it inf}, \varepsilon_\text{\it sup}\right[$ and that $\Psi_2^R$ is a decreasing bijection. Lastly, theorem~6.1 in~\cite{godlewski96} states that for $\|\bm{U}_R - \bm{U}_L\|$ sufficiently small, the solution of (\ref{SystHypVect})-(\ref{SystCI}) is unique. Similarly to theorem~7.1 in~\cite{godlewski96}, we get a condition on the initial data which ensures the existence of the solution.
	
	\begin{Theorem} \label{thm:Intersect}
		If the characteristic fields are genuinely nonlinear (\ref{NLGen}), then the solution to the Riemann problem (\ref{SystHypVect})-(\ref{SystCI}) exists and is unique, provided that
		\begin{equation}
			\underset{\varepsilon\rightarrow \varepsilon_\textit{inf}+}{\lim} \Psi_2^R(\varepsilon) - \Phi_1^L(\varepsilon) > 0
			\qquad\text{and}\qquad
			\underset{\varepsilon\rightarrow \varepsilon_\textit{sup}-}{\lim} \Psi_2^R(\varepsilon) - \Phi_1^L(\varepsilon) < 0\, ,
			\label{IntersectCond}
		\end{equation}
		with $\Phi_1^L$ and $\Psi_2^R$ given in (\ref{WNLGen}).
	\end{Theorem}
	
	\begin{proof}
		To ensure that the solution described above exists, the forward and backward wave curves $\Phi_1^L$ and $\Psi_2^R$ must intersect at a strain $\varepsilon_M$ satisfying (\ref{IntermediateState}). The associated functions are continuous bijections over the interval $\left]\varepsilon_\textit{inf}, \varepsilon_\textit{sup}\right[$. Moreover, $\Phi_1^L$ is strictly increasing while $\Psi_2^R$ is strictly decreasing. Therefore, they intersect once over $\left]\varepsilon_\textit{inf}, \varepsilon_\text{\it sup}\right[$ if and only if their ranges intersect. The latter are respectively
		\begin{equation*}
			\left]\underset{\varepsilon\rightarrow \varepsilon_\textit{inf}+}{\lim} \Phi_1^L(\varepsilon), \underset{\varepsilon\rightarrow \varepsilon_\textit{sup}-}{\lim} \Phi_1^L(\varepsilon)\right[
			\quad\text{and}\quad
			\left] \underset{\varepsilon\rightarrow \varepsilon_\textit{sup}-}{\lim} \Psi_2^R(\varepsilon), \underset{\varepsilon\rightarrow \varepsilon_\textit{inf}+}{\lim} \Psi_2^R(\varepsilon) \right[ .
		\end{equation*}
		A comparison between these bounds ends the proof (\ref{IntersectCond}).
		\qed
	\end{proof}
	
	Theorem~\ref{thm:Intersect} can be written in terms of $v_R$. Indeed, (\ref{IntersectCond}) is equivalent to
	\begin{equation}
		\Phi_2^\textit{inf}(\varepsilon_R)
		< v_R <
		\Phi_2^\textit{sup}(\varepsilon_R)\, ,
		\label{IntersectCondJump1}
	\end{equation}
	where
	\begin{equation}
	{\addtolength{\jot}{0.2em}
		\begin{aligned}
		&\Phi_2^\textit{inf}(\varepsilon_R)=\underset{\varepsilon\rightarrow \varepsilon_\textit{inf}+}{\lim} \Phi_1^L(\varepsilon) + v_R - \Psi_2^R(\varepsilon)\, ,\\
		&\Phi_2^\textit{sup}(\varepsilon_R)=\underset{\varepsilon\rightarrow \varepsilon_\textit{sup}-}{\lim} \Phi_1^L(\varepsilon) + v_R - \Psi_2^R(\varepsilon)\, .
		\end{aligned}}
	\label{IntersectCondJump0}
	\end{equation}
	The functions $\Phi_2^\textit{inf}$ and $\Phi_2^\textit{sup}$ in (\ref{IntersectCondJump1}) are the forward wave curves passing through the states $\bm{U}_{\textit{inf}}$ and $\bm{U}_{\textit{sup}}$ respectively, such that
	\begin{equation}
		\bm{U}_{\textit{inf}} = \underset{\varepsilon\rightarrow \varepsilon_\textit{inf}+}{\lim} \left( \varepsilon, \Phi_1^L(\varepsilon) \right)^\top
		\qquad\text{and}\qquad
		\bm{U}_{\textit{sup}} = \underset{\varepsilon\rightarrow \varepsilon_\textit{sup}-}{\lim} \left( \varepsilon, \Phi_1^L(\varepsilon) \right)^\top .
		\label{IntersectCondJump2}
	\end{equation}
	Graphically, $\Phi_2^\textit{inf}$ and $\Phi_2^\textit{sup}$ correspond to the dashed curve $\Phi_2^M$ on figure~\ref{fig:AdmissibleConcave}-(a) when $\varepsilon_M$ tends towards $\varepsilon_\textit{inf}$ or $\varepsilon_\textit{sup}$ respectively. Since the curve $\Phi_1^L$ is only translated vertically when $v_L$ varies, the condition (\ref{IntersectCondJump1})-(\ref{IntersectCondJump2}) can be written in terms of the velocity jump $v_R - v_L$ by substracting $v_L$ in (\ref{IntersectCondJump1}). For analytical expressions and remarks, see (\ref{IntersectCondNLGen}) in appendix~\ref{subsec:MathsJump}.
	
	\begin{figure}
		\centering
		\begin{minipage}{0.49\linewidth}
			\vspace{-1.9em}
			\centering
			(a)
			
			\vspace{1.2em}
			\includegraphics{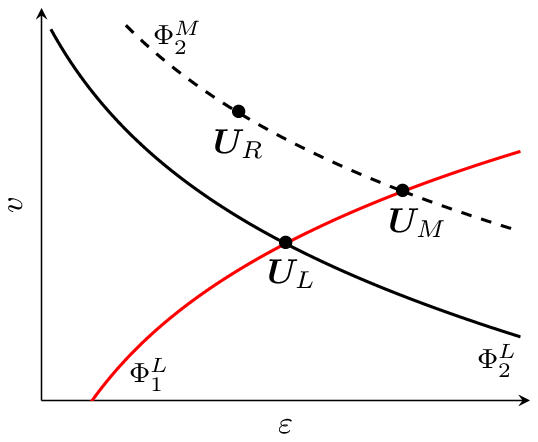}
		\end{minipage}
		\begin{minipage}{0.49\linewidth}
			\centering
			(b)
			
			\vspace{0.8em}
			\includegraphics{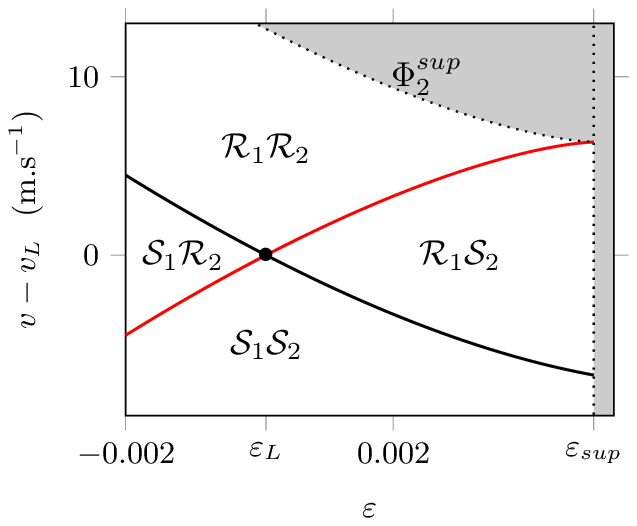}
		\end{minipage}
		\caption{(a) Construction of the solution to (\ref{IntermediateState2}). (b) Admissibility regions and hyperbolicity domain (white) for Landau's law (\ref{BLawLandau}) with $\varepsilon_L=10^{-4}$.\label{fig:AdmissibleConcave}}
	\end{figure}
	
	In (\ref{IntersectCondJump1}), $\Phi_2^\textit{inf}(\varepsilon_R)$ is infinite if $\varepsilon_\textit{inf} = {-\infty}$ or if $\sigma(\varepsilon)$ tends towards ${-\infty}$ when $\varepsilon$ tends towards $\varepsilon_\textit{inf}+$. The value of $\Phi_2^\textit{sup}(\varepsilon_R)$ is infinite if $C(\varepsilon)$ tends towards ${+\infty}$ when $\varepsilon$ tends towards $\varepsilon_\textit{sup}-$. If both are infinite, then theorem~\ref{thm:Intersect} is satisfied for every initial data. Else, there exists a bound on $v_R - v_L$ which ensures the existence of the solution. This result is new and is not known in the literature.
	
	Now, we describe the admissibility regions, i.e. the regions of the $\varepsilon$-$v$ plane where a given wave structure is admissible given $\bm{U}_L$. This is similar to the approach presented in theorem~7.1 of \cite{godlewski96}. Thus, we draw the forward wave curves $\Phi_1^L$ and $\Phi_2^L$ passing through $\bm{U}_L$. These curves divide the plane into four regions (figure~\ref{fig:AdmissibleConcave}-(a)). When $\bm{U}_M$ belongs to $\Phi_1^L$, (\ref{WNLGen}) states which kind of 1-wave connects $\bm{U}_M$ to $\bm{U}_L$. Then, we draw the forward wave curve $\Phi_2^M$ passing through $\bm{U}_M$. For any $\bm{U}_R$ belonging to $\Phi_2^M$, we know which kind of 2-wave connects it to $\bm{U}_M$ (\ref{WNLGen}). Finally, we obtain a map of the admissible combinations of 1-waves and 2-waves (figure~\ref{fig:AdmissibleConcave}-(b)). If (\ref{IntersectCondJump1}) is satisfied, then four regions are distinguished:
	\begin{itemize}
		\item If $v_R\geqslant\Phi_1^L(\varepsilon_R)$ and $v_R\geqslant\Phi_2^L(\varepsilon_R)$, region $\mathcal{R}_1\mathcal{R}_2$,
		
		\item Else, if $v_R\geqslant\Phi_1^L(\varepsilon_R)$ and $v_R<\Phi_2^L(\varepsilon_R)$, region $\mathcal{S}_1\mathcal{R}_2$,
		
		\item Else, if $v_R\geqslant\Phi_2^L(\varepsilon_R)$ and $v_R<\Phi_1^L(\varepsilon_R)$, region $\mathcal{R}_1\mathcal{S}_2$,
		
		\item Else, region $\mathcal{S}_1\mathcal{S}_2$.
	\end{itemize}
	
	\paragraph*{Model 1 {\mdseries (hyperbola)}.}
	Here, $\left]\varepsilon_\textit{inf}, \varepsilon_\textit{sup}\right[=\left]{-d},{+\infty}\right[$. The limit of $\sigma(\varepsilon)$ when $\varepsilon$ tends towards ${-d}$ is equal to ${-\infty}$. Also, the limit of $C(\varepsilon)$ when $\varepsilon$ tends towards ${+\infty}$ is equal to ${+\infty}$. Therefore, theorem~\ref{thm:Intersect} is satisfied for every left and right states in $\left]\varepsilon_\textit{inf}, \varepsilon_\text{\it sup}\right[$. The computation of the solution is detailed in section~\ref{sec:NumEx}, for a configuration with two shocks and another configuration with two rarefactions.
	
	\paragraph*{Model 3 {\mdseries (Landau)}.}
	Here, $\delta=0$ in (\ref{BLawLandau}) and $\left]\varepsilon_\textit{inf}, \varepsilon_\text{\it sup}\right[=\left]{-\infty},1/2\beta\right[$. At the lower edge, $\varepsilon_\textit{inf} = {-\infty}$. But at the upper edge, $C(\varepsilon)$ vanishes when $\varepsilon$ tends towards $1/2\beta$. Therefore, theorem~\ref{thm:Intersect} is not satisfied for high values of the velocity jump. To illustrate, we take $\varepsilon_L={-\varepsilon_R}={10}^{-4}$ and the parameters issued from table~\ref{tab:Params}. Condition (\ref{IntersectCondNLGen}) then becomes $v_R-v_L \leqslant\; 13.07$~m.s$^{-1}$. A graphical interpretation is given on figure~\ref{fig:AdmissibleConcave}-(b).
	
	\subsection{Convex-concave constitutive laws}\label{subsec:ConvexConcave}
	
	Let us assume that $\sigma''$ is strictly decreasing and equals zero at $\varepsilon=\varepsilon_0$. Therefore, the characteristic fields are neither linearly degenerate nor genuinely nonlinear. The stress function $\varepsilon\mapsto \sigma(\varepsilon)$ is strictly convex for $\varepsilon<\varepsilon_0$ and strictly concave for $\varepsilon>\varepsilon_0$.
	For any $a$ and $b$, let us denote
	\begin{equation*}
		\begin{aligned}
			&a \underset{L}{\leqslant} b \qquad\Leftrightarrow\qquad (\varepsilon_L - \varepsilon_0)\, a \leqslant (\varepsilon_L - \varepsilon_0)\, b \, ,\\
			&a \underset{R}{<} b \qquad\Leftrightarrow\qquad (\varepsilon_R - \varepsilon_0)\, a < (\varepsilon_R - \varepsilon_0)\, b \, .
		\end{aligned}
	\end{equation*}
	Similar notations are used for other kinds of inequalities, such as $a \underset{L}{>} b$, etc. From the graphical method in section~\ref{subsec:GraphMeth} based on convex hull constructions, forward and backward wave curves write
	\begin{equation}
	\begin{aligned}
	&\Phi_1^L(\varepsilon) =
	\left\lbrace
	{ %\addtolength{\jot}{0.1em}
		\begin{aligned}
			&\mathcal{S}_1^L(\varepsilon) & &\text{if }
			\varepsilon_L^* \,\underset{L}{\leqslant}\, \varepsilon \,\underset{L}{<}\, \varepsilon_L \, ,\\
			&\mathcal{R}_1^L(\varepsilon) & &\text{if }
			\varepsilon \,\underset{L}{\geqslant}\, \varepsilon_L \, ,\\
			&\mathcal{SR}_1^L(\varepsilon) & &\text{if }
			\varepsilon \,\underset{L}{<}\, \varepsilon_L^* \, ,
		\end{aligned}}
		\right.
	\Phi_2^L(\varepsilon) =
	\left\lbrace
	{%\addtolength{\jot}{0.1em}
		\begin{aligned}
		&\mathcal{S}_2^L(\varepsilon) & &\text{if }
		\varepsilon \,\underset{L}{<}\, {\varepsilon_L}^\dagger
		\quad\text{or}\quad
		\varepsilon \,\underset{L}{>}\, \varepsilon_L
		\, ,\\
		&\mathcal{R}_2^L(\varepsilon) & &\text{if }
		\varepsilon_0 \,\underset{L}{\leqslant}\, \varepsilon \,\underset{L}{\leqslant}\, \varepsilon_L \, ,\\
		&\mathcal{RS}_2^L(\varepsilon) & &\text{if }
		{\varepsilon_L}^\dagger \,\underset{L}{\leqslant}\, \varepsilon \,\underset{L}{<}\, \varepsilon_0 \, ,
		\end{aligned}}
	\right.\\
	&\Psi_1^R(\varepsilon) =
		\left\lbrace
		{%\addtolength{\jot}{0.1em}
		\begin{aligned}
			&\mathcal{S}_1^R(\varepsilon) & &\text{if }
			\varepsilon \,\underset{R}{<}\, {\varepsilon_R}^\dagger
			\quad\text{or}\quad
			\varepsilon \,\underset{R}{>}\, \varepsilon_R
			\, ,\\
			&\mathcal{R}_1^R(\varepsilon) & &\text{if }
			\varepsilon_0 \,\underset{R}{\leqslant}\, \varepsilon \,\underset{R}{\leqslant}\, \varepsilon_R \, ,\\
			&\mathcal{RS}_1^R(\varepsilon) & &\text{if }
			{\varepsilon_R}^\dagger \,\underset{R}{\leqslant}\, \varepsilon \,\underset{R}{<}\, \varepsilon_0 \, ,
		\end{aligned}}
	\right.
	\Psi_2^R(\varepsilon) =
		\left\lbrace
		{%\addtolength{\jot}{0.1em}
		\begin{aligned}
			&\mathcal{S}_2^R(\varepsilon) & &\text{if }
			\varepsilon_R^* \,\underset{R}{\leqslant}\, \varepsilon \,\underset{R}{<}\, \varepsilon_R \, ,\\
			&\mathcal{R}_2^R(\varepsilon) & &\text{if }
			\varepsilon \,\underset{R}{\geqslant}\, \varepsilon_R \, ,\\
			&\mathcal{SR}_2^R(\varepsilon) & &\text{if }
			\varepsilon \,\underset{R}{<}\, \varepsilon_R^* \, .
		\end{aligned}}
		\right.
	\end{aligned}
	\label{WNLNGen}
	\end{equation}
	When $\varepsilon_0\rightarrow {-\infty}$, the constitutive law $\sigma(\varepsilon)$ becomes strictly concave. In this case, $\varepsilon$, $\varepsilon_L$ and $\varepsilon_R$ are always higher than $\varepsilon_0$. Thus, \raisebox{0.3em}{$\underset{L}{<}$} can be replaced by $<$ in (\ref{WNLNGen}) (idem for similar notations). Moreover, $\varepsilon_L^*$, $\varepsilon_R^*$, ${\varepsilon_L}^\dagger$ and ${\varepsilon_R}^\dagger$ tend towards ${-\infty}$. Therefore, we recover the wave curves (\ref{WNLGen}).
	
	Forward and backward wave curves are Lipschitz continuous and they are $C^2$ in the vicinity of the states $\bm{U}_L$ or $\bm{U}_R$. Their regularity may be reduced to $C^1$ after the first crossing with the line $\varepsilon=\varepsilon_0$ (sections 9.3 to 9.5 of \cite{dafermos05}). From the properties of each elementary curve studied before, we deduce that $\Phi_1^L$ is an increasing bijection over $\left]\varepsilon_\textit{inf}, \varepsilon_\text{\it sup}\right[$ and $\Psi_2^R$ a decreasing bijection. Lastly, theorem~9.5.1 in \cite{dafermos05} states that for $\|\bm{U}_R - \bm{U}_L\|$ sufficiently small, the solution is unique. Similarly to theorem~\ref{thm:Intersect}, we deduce a condition which ensures the existence of the solution for any initial data.
	
	\begin{Theorem} \label{thm:IntersectNLNGen}
		If the constitutive law is strictly convex for $\varepsilon<\varepsilon_0$ and strictly concave for $\varepsilon>\varepsilon_0$, then the solution to the Riemann problem (\ref{SystHypVect})-(\ref{SystCI}) exists and is unique, provided that
		\begin{equation}
			\underset{\varepsilon\rightarrow \varepsilon_\textit{inf}+}{\lim} \Psi_2^R(\varepsilon) - \Phi_1^L(\varepsilon) > 0
			\qquad\text{and}\qquad
			\underset{\varepsilon\rightarrow \varepsilon_\textit{sup}-}{\lim} \Psi_2^R(\varepsilon) - \Phi_1^L(\varepsilon) < 0 \, ,
			\label{IntersectCondNLNGenThm}
		\end{equation}
		with $\Phi_1^L$ and $\Psi_2^R$ given in (\ref{WNLNGen}).
	\end{Theorem}
	
	\begin{proof}
		Similarly to theorem~\ref{thm:Intersect}, we can reduce the existence criterion to a comparison between the ranges of $\Phi_1^L$ and $\Psi_2^R$.
		\qed
	\end{proof}
	
	Theorem~\ref{thm:IntersectNLNGen} can be written in terms of the velocity jump $v_R-v_L$. The analytical expressions (\ref{IntersectCondNLNGen1})-(\ref{IntersectCondNLNGen4}) are given in appendix~\ref{subsec:MathsJump}. If both limits of $C(\varepsilon)$ are infinite when $\varepsilon$ tends towards $\varepsilon_\textit{inf}+$ or $\varepsilon_\textit{sup}-$, then (\ref{IntersectCondNLNGenThm}) is satisfied for every initial data. Else, there exists a bound on the velocity jump, which ensures the existence of the solution.
	
	\paragraph*{Case $\varepsilon_L = \varepsilon_0$.}
	We describe the admissibility regions when the left state is on the inflexion point. As we did for concave constitutive laws, we draw the forward wave curve $\Phi_1^L$ passing through $\bm{U}_L$ (figure~\ref{fig:AdmissibleLandau3}-(a)). Let us consider an intermediate state $\bm{U}_M$ belonging to $\Phi_1^L$. It is connected to $\bm{U}_L$ through a 1-rarefaction (\ref{WNLNGen}). Then, we draw the forward wave curve $\Phi_2^M$ passing through $\bm{U}_M$. For any $\bm{U}_R$ belonging to $\Phi_2^M$, one knows which kind of 2-wave connects $\bm{U}_M$ to $\bm{U}_R$ (\ref{WNLNGen}). On figure~\ref{fig:AdmissibleLandau3}-(a), $\varepsilon_M > \varepsilon_0$. Therefore, we have a 2-shock if $\varepsilon_R< {\varepsilon_M}^\dagger$ or $\varepsilon_R> \varepsilon_M$, a 2-rarefaction if $\varepsilon_0 \leqslant \varepsilon_R \leqslant \varepsilon_M$ and a 2-rarefaction-shock else. Here, the 2-wave is a rarefaction-shock.
	
	\begin{figure}
		\begin{minipage}{0.49\linewidth}
			\centering
			(a)
			
			\vspace{0.9em}
			\includegraphics{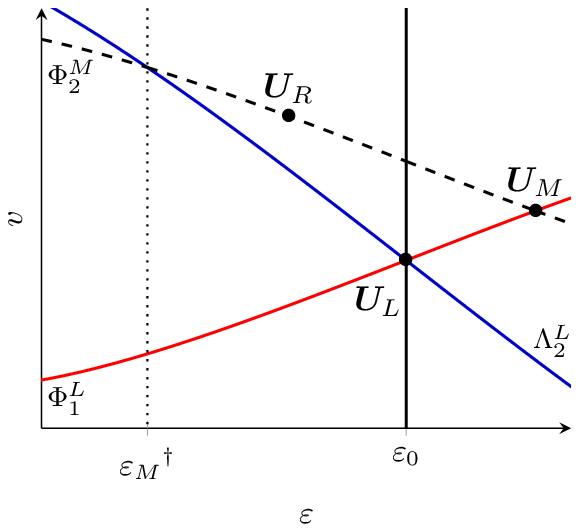}
		\end{minipage}
		\begin{minipage}{0.49\linewidth}
			\centering
			(b)
			
			\vspace{0.4em}
			\includegraphics{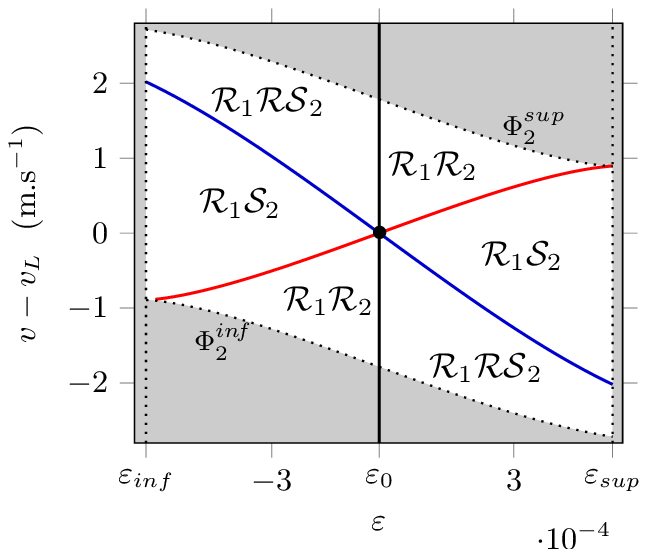}
		\end{minipage}
		
		\caption{Case $\varepsilon_L=\varepsilon_0$. (a) Construction of the solution to (\ref{IntermediateState2}). (b) Admissibility regions and hyperbolicity domain (white) for Landau's law (\ref{BLawLandau}) with a cubic nonlinearity and parameters from table~\ref{tab:Params}.\label{fig:AdmissibleLandau3}}
	\end{figure}
	
	To achieve the partition of the $\varepsilon$-$v$ space into admissibility regions, we introduce the curve $\Lambda_2^L$ which marks the equality case in Liu's entropy condition for 2-shocks (\ref{Liu}). The curve $\Lambda_2^L$ marks the frontier between the admissibility regions of 2-shocks and 2-rarefaction-shocks. It is the set of right states $\bm{u}$ belonging to $\Phi_2^M$ such that $\varepsilon={\varepsilon_M}^\dagger$, or equivalently $\varepsilon^*=\varepsilon_M$, when $\bm{U}_M$ varies along $\Phi_1^L$ (figure~\ref{fig:AdmissibleLandau3}-(a)). Hence, $\bm{u}$ satisfies $v = \mathcal{RS}_2^M(\varepsilon)$, where $\bm{U}_M = (\varepsilon^*,\Phi_1^L(\varepsilon^*))^\top$:
	\begin{equation}
		v = \Phi_1^L(\varepsilon^*) + c(\varepsilon^*) (\varepsilon^* - \varepsilon) \equiv \Lambda_2^L(\varepsilon) \, .
		\label{Lambda2L}
	\end{equation}
	Finally, we obtain a map of the admissible combinations of 1-waves and 2-waves (figure~\ref{fig:AdmissibleLandau3}-(b)). If (\ref{IntersectCondJump1}) is satisfied and $\varepsilon_L=\varepsilon_0$, then three regions are distinguished:
	\begin{itemize}
		\item If $v_R \,\underset{R}{\geqslant}\, \Phi_1^L(\varepsilon_R)$, region $\mathcal{R}_1\mathcal{R}_2$,
		\item Else, if $v_R \,\underset{R}{\geqslant}\, \Lambda_2^L(\varepsilon_R)$, region $\mathcal{R}_1\mathcal{S}_2$,
		\item Else, region $\mathcal{R}_1\mathcal{RS}_2$.
	\end{itemize}
	
	\paragraph*{Case $\varepsilon_L \neq \varepsilon_0$.}
	Figure~\ref{fig:AdmissibleLandau3Gen} represents the admissibility regions for $\varepsilon_L>\varepsilon_0$. Similarly, figure~\ref{fig:AdmissibleLandau3Gen2} shows the admissibility regions for $\varepsilon_L>\varepsilon_0$. In both cases, we draw the forward wave curves $\Phi_1^L$ and $\Phi_2^L$ passing through $\bm{U}_L$. For any intermediate state $\bm{U}_M$ belonging to $\Phi_1^L$, equation (\ref{WNLNGen}) selects the 1-wave which connects $\bm{U}_M$ to $\bm{U}_L$: a 1-shock if $\varepsilon_L^*\leqslant \varepsilon_M < \varepsilon_L$, a 1-rarefaction if $\varepsilon_M\geqslant \varepsilon_L$ and a 1-shock-rarefaction else. Then, we draw the curve $\Lambda_2^L$ marking Liu's condition for 2-shocks. Thus, we can already qualify six admissibility regions.
	
	\begin{figure}
		\centering
		\includegraphics{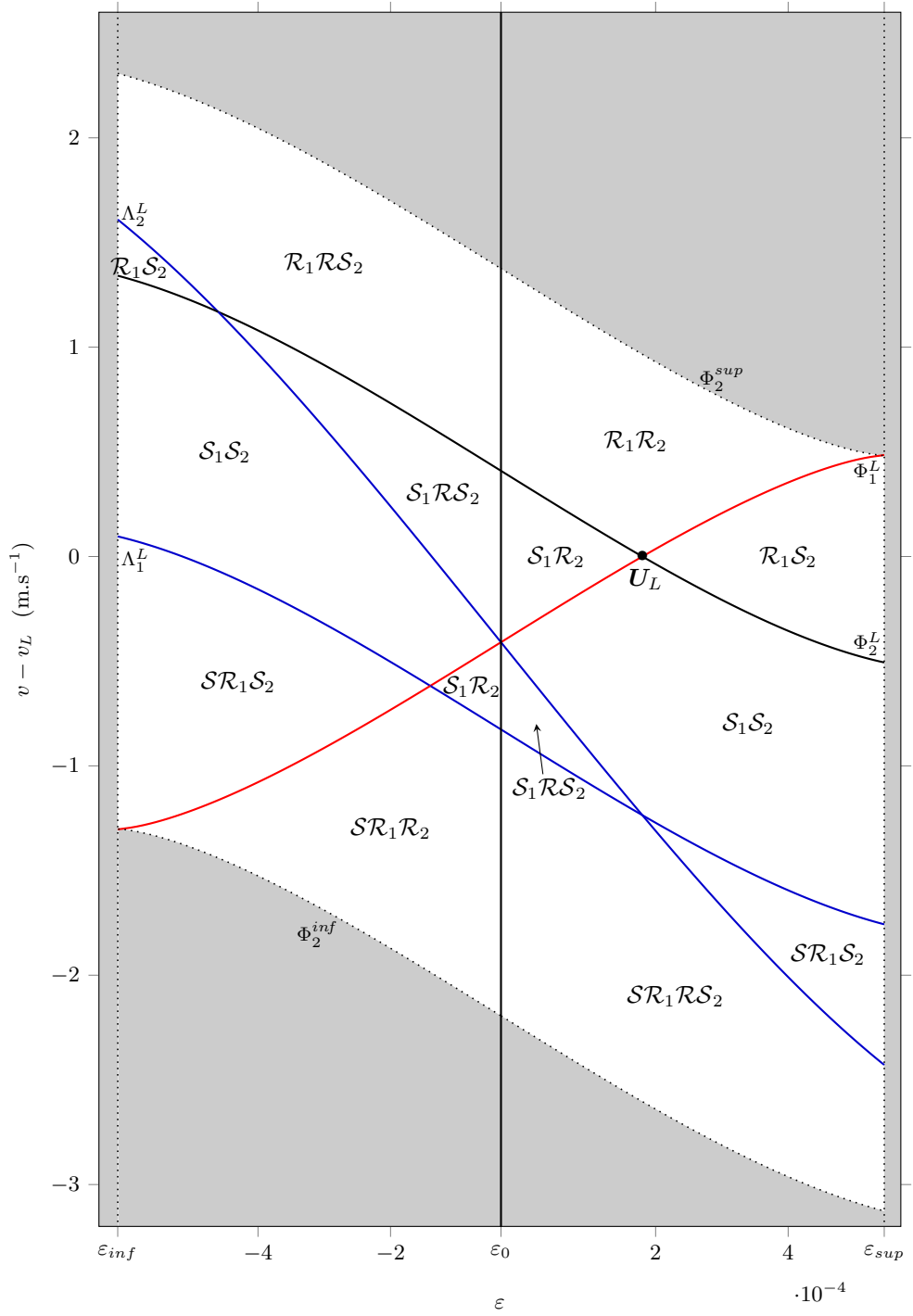}
		
		\caption{Case $\varepsilon_L>\varepsilon_0$. Admissibility regions and hyperbolicity domain (white) for Landau's law (\ref{BLawHyp}) with a cubic nonlinearity and the parameters from table~\ref{tab:Params}. Here, $\varepsilon_L=1.8\times 10^{-4}$.\label{fig:AdmissibleLandau3Gen}}
	\end{figure}
	
	\begin{figure}
		\centering
		\includegraphics{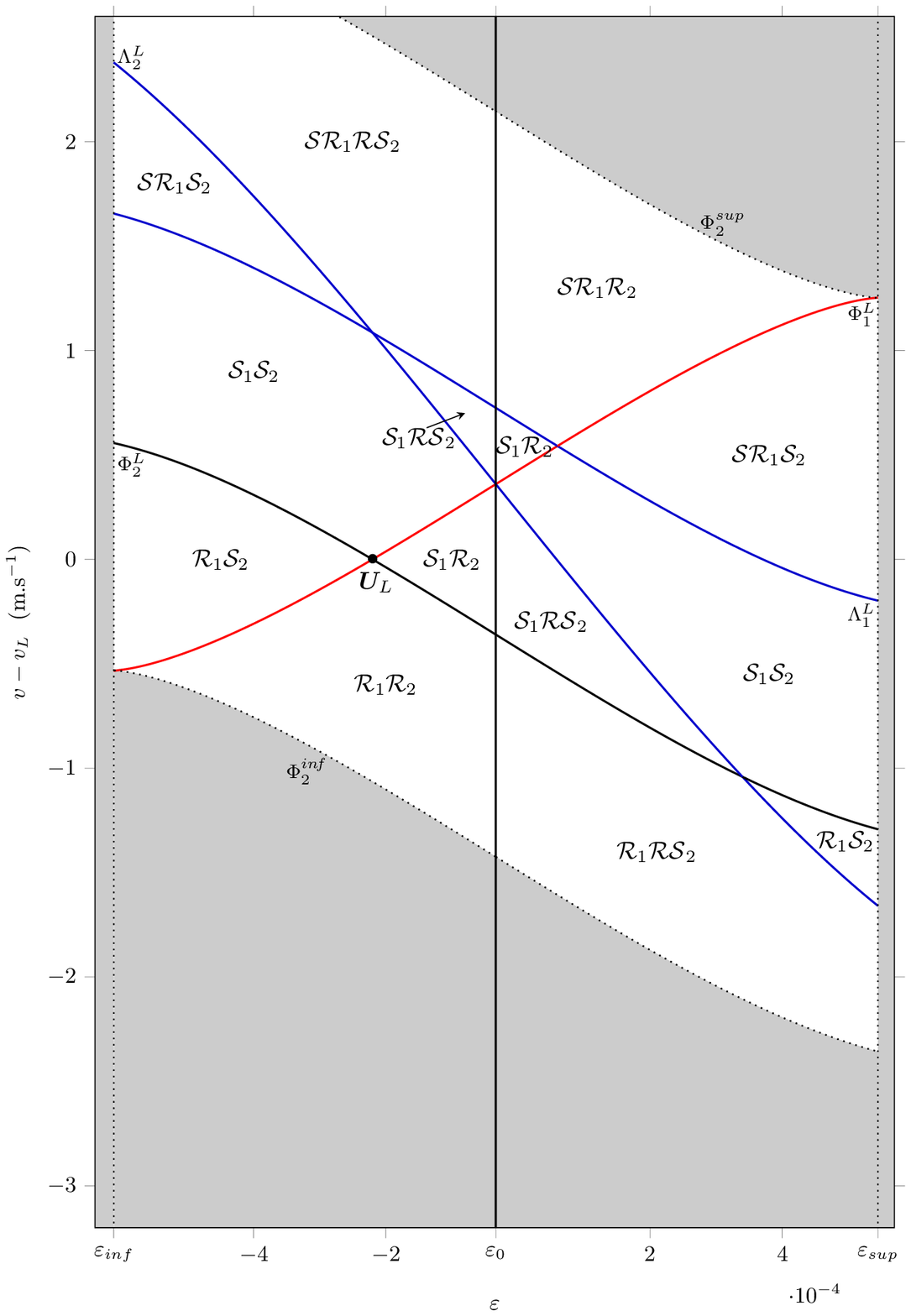}
		
		\caption{Case $\varepsilon_L<\varepsilon_0$. Same as figure~\ref{fig:AdmissibleLandau3Gen}, but with $\varepsilon_L={-2.2}\times 10^{-4}$.\label{fig:AdmissibleLandau3Gen2}}
	\end{figure}
	
	To achieve the partition of the $\varepsilon$-$v$ space, we introduce the curve $\Lambda_1^L$ which corresponds to the equality case in Liu's entropy condition for 1-shocks (\ref{Liu}). The curve $\Lambda_1^L$ marks the frontier between the admissibility regions of 1-shocks and 1-shock-rarefactions. It is the locus of right states $\bm{u}$ belonging to $\Phi_2^M$, where the intermediate state is $\bm{U}_M=(\varepsilon_L^*, \Phi_1^L(\varepsilon_L^*))^\top$. Since $(\varepsilon_L^*-\varepsilon_0)(\varepsilon_L-\varepsilon_0)\leqslant 0$, the inequalities depending on $\varepsilon_M-\varepsilon_0$ in $\Phi_2^M(\varepsilon)$ (\ref{WNLNGen}) can be changed in inequalities depending on $\varepsilon_L-\varepsilon_0$. Hence,
	\begin{equation}
		v
		=
		\left\lbrace
		{\addtolength{\jot}{0.1em}
			\begin{aligned}\!
			&\mathcal{S}_2^M(\varepsilon) & &\text{if }
			\varepsilon \,\underset{L}{>}\, \varepsilon_L 
			\quad\text{or}\quad
			\varepsilon \,\underset{L}{<}\, \varepsilon_L^*\\
			&\mathcal{R}_2^M(\varepsilon) & &\text{if }
			\varepsilon_0 \,\underset{L}{\geqslant}\, \varepsilon \,\underset{L}{\geqslant}\, \varepsilon_L^* \\
			&\mathcal{RS}_2^M(\varepsilon) & &\text{if }
			\varepsilon_L \,\underset{L}{\geqslant}\, \varepsilon \,\underset{L}{>}\, \varepsilon_0
			\end{aligned}}
		\right.
		\equiv \Lambda_1^L(\varepsilon) \, .
		\label{Lambda1L}
	\end{equation}
	Finally, if (\ref{IntersectCondJump1}) is satisfied and $\varepsilon_L\neq\varepsilon_0$, then nine regions are distinguished:
	\begin{itemize}
		\item If $v_R \,\underset{L}{\geqslant}\, \Phi_2^L(\varepsilon_R)$, $v_R \,\underset{L}{\geqslant}\, \Phi_1^L(\varepsilon_R)$ and $\varepsilon_R \,\underset{L}{\geqslant}\, \varepsilon_0$, region $\mathcal{R}_1\mathcal{R}_2$,
		\item Else, if $v_R \,\underset{L}{\geqslant}\, \Phi_2^L(\varepsilon_R)$ and $\left[v_R \,\underset{L}{<}\, \Phi_1^L(\varepsilon_R) \text{ or } v_R \,\underset{L}{\leqslant}\, \Lambda_2^L(\varepsilon_R)\right]$, region $\mathcal{R}_1\mathcal{S}_2$,
		\item Else, if $v_R \,\underset{L}{>}\, \Phi_2^L(\varepsilon_R)$, $v_R \,\underset{L}{>}\, \Lambda_2^L(\varepsilon_R)$ and $\varepsilon_R \,\underset{L}{<}\, \varepsilon_0$, region $\mathcal{R}_1\mathcal{RS}_2$.
		\item Else, if $v_R \,\underset{L}{<}\, \Phi_2^L(\varepsilon_R)$, $v_R \,\underset{L}{\geqslant}\, \Lambda_1^L(\varepsilon_R)$ and $v_R \,\underset{R}{\geqslant}\, \Phi_1^L(\varepsilon_R)$, region $\mathcal{S}_1\mathcal{R}_2$,
		\item Else, if $v_R \,\underset{L}{\leqslant}\, \Phi_2^L(\varepsilon_R)$, $v_R \,\underset{L}{\geqslant}\, \Lambda_1^L(\varepsilon_R)$ and $v_R \,\underset{R}{<}\, \Lambda_2^L(\varepsilon_R)$, region $\mathcal{S}_1\mathcal{RS}_2$.
		\item Else, if $v_R \,\underset{L}{\leqslant}\, \Phi_1^L(\varepsilon_R)$, $v_R \,\underset{L}{\leqslant}\, \Lambda_1^L(\varepsilon_R)$ and $\varepsilon_R \,\underset{L}{\leqslant}\, \varepsilon_0$, region $\mathcal{SR}_1\mathcal{R}_2$,
		\item Else, if $v_R \,\underset{L}{\leqslant}\, \Lambda_1^L(\varepsilon_R)$ and $\left[v_R \,\underset{L}{>}\, \Phi_1^L(\varepsilon_R) \text{ or } v_R \,\underset{L}{\geqslant}\, \Lambda_2^L(\varepsilon_R)\right]$, region $\mathcal{SR}_1\mathcal{S}_2$,
		\item Else, if $v_R \,\underset{L}{<}\, \Lambda_1^L(\varepsilon_R)$, $v_R \,\underset{L}{<}\, \Lambda_2^L(\varepsilon_R)$ and $\varepsilon_R \,\underset{L}{>}\, \varepsilon_0$, region $\mathcal{SR}_1\mathcal{RS}_2$.
		\item Else, region $\mathcal{S}_1\mathcal{S}_2$.
	\end{itemize}
	
	\paragraph*{Model 2 {\mdseries (tanh)}.}
		Here, $\left]\varepsilon_\textit{inf}, \varepsilon_\text{\it sup}\right[=\mathbb{R}$. The limit of $C(\varepsilon)$ when $\varepsilon$ tends towards ${\pm\infty}$ is equal to ${\pm \frac{\pi}{2}} c_0\, d$. Therefore, the velocity jump is always bounded. This property is illustrated on figure~\ref{fig:AdmissibleTanh}. If $\varepsilon_L={-\varepsilon_R}=10^{-4}$, the velocity jump must satisfy $|v_R-v_L|\leqslant 6.16~\text{m.s}^{-1}$.
		
	\begin{figure}
		\centering
		\includegraphics{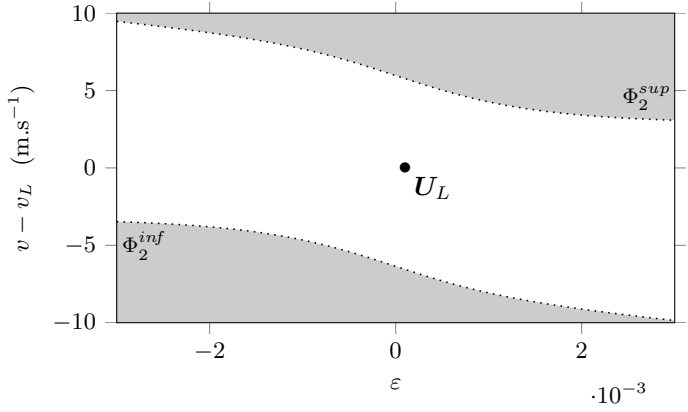}
		
		\caption{Existence domain for the tanh constitutive law with parameters from table~\ref{tab:Params} and $\varepsilon_L=10^{-4}$.\label{fig:AdmissibleTanh}}
	\end{figure}
	
	\paragraph*{Model 3 {\mdseries (Landau)}.}
		Here, $\left]\varepsilon_\textit{inf}, \varepsilon_\text{\it sup}\right[$ is bounded (\ref{OmegaLandau}). The limit of $C(\varepsilon)$ when $\varepsilon$ tends towards $\varepsilon_\textit{sup}$ or $\varepsilon_\textit{inf}$ is equal to ${\pm \frac{\pi}{2}} c_0\, \frac{\beta^2+3\delta}{6\delta\sqrt{3\delta}}$. Therefore, the velocity jump $v_R-v_L$ is also bounded, which is illustrated on figures \ref{fig:AdmissibleLandau3} and \ref{fig:AdmissibleLandau3Gen}. With the parameters from table~\ref{tab:Params}, it must belong to $\left[{-1.91}, 1.66\right]~\text{m.s}^{-1}$ if $\varepsilon_L={-\varepsilon_R}=10^{-4}$. The computation of the solution is detailed in section~\ref{sec:NumEx}, for a configuration with two compound waves.
	
	%%%%%%%%%%%%%%%%%%%%%%%%%%%%%%%%%%%%%%%%%%%%%%%%%%%%%%%%%%%%%%%%%%%
	%-------------------------- illustr --------------------------------
	
	\section{Numerical examples}\label{sec:NumEx}
	
	With the parameters issued from table~\ref{tab:Params}, we give two examples for the hyperbola constitutive law (\ref{BLawHyp}) and one for Landau's law (\ref{BLawLandau}).
	
	\paragraph*{1-shock, 2-shock \mdseries{(hyperbola)}.}
	On figure~\ref{fig:Hyp1}, we display the solution with initial data $\varepsilon_L=-10^{-4}$, $\varepsilon_R=10^{-4}$, $v_L=0.5$~m.s$^{-1}$ and $v_R=-0.5$~m.s$^{-1}$. The solution consists of two shocks:
	\begin{equation}
	\bm{U}(x,t)=
	\left\lbrace\!
	{\renewcommand{\arraystretch}{1.2}
		\begin{array}{ll}
		\bm{U}_L &\text{if } x < s_1\, t\, ,\\
		\bm{U}_M &\text{if } s_1\, t < x < s_2\, t\, ,\\
		\bm{U}_R &\text{if } s_2\, t < x\, .\\
		\end{array}}\!\right.
	\end{equation}
	Here $\varepsilon_M\approx -2.282\times 10^{-4}$. Therefore, the shock speeds are $s_1\approx -2.353$~km/s and $s_2\approx 2.129$~km/s (\ref{EqRH2}).
	
	\paragraph*{1-rarefaction, 2-rarefaction \mdseries{(hyperbola)}.}
	On figure~\ref{fig:Hyp4}, we represent the solution with initial data $\varepsilon_L=-10^{-4}$, $\varepsilon_R=10^{-4}$, $v_L=-0.5$~m.s$^{-1}$ and $v_R=0.5$~m.s$^{-1}$. It consists of two rarefactions:
	\begin{equation}
	\bm{U}(x,t)=
	\left\lbrace\!
	{\renewcommand{\arraystretch}{1.2}
		\begin{array}{ll}
		\bm{U}_L &\text{if } x \leqslant {-c}(\varepsilon_L)\, t\, ,\\
		{\bm V}_1(x/t) &\text{if } {-c}(\varepsilon_L)\, t \leqslant x \leqslant {-c}(\varepsilon_M)\, t\, ,\\
		\bm{U}_M &\text{if } {-c}(\varepsilon_M)\, t\leqslant x \leqslant c(\varepsilon_M)\, t\, ,\\
		{\bm V}_2(x/t) &\text{if } c(\varepsilon_M)\, t \leqslant x \leqslant c(\varepsilon_R)\, t\, ,\\
		\bm{U}_R &\text{if } c(\varepsilon_R)\, t \leqslant x\, ,\\
		\end{array}}\!\right.
	\end{equation}
	where ${\bm V}_1(\xi)$ and ${\bm V}_2(\xi)$ satisfy (\ref{SolRarefaction}) with $p=1$ and $p=2$ respectively. Here, $\varepsilon_M\approx 2.839\times 10^{-4}$.
	
	\paragraph*{1-shock-rarefaction, 2-rarefaction-shock \mdseries{(Landau)}.}
	On figure~\ref{fig:Lan1}, we display the solution with initial data $\varepsilon_L=-10^{-4}$, $\varepsilon_R=-2\times 10^{-4}$, $v_L=-0.6$~m.s$^{-1}$ and $v_R=0.6$~m.s$^{-1}$. It consists of two compound waves:
	\begin{equation}
	\bm{U}(x,t)=
	\left\lbrace\!
	{\renewcommand{\arraystretch}{1.2}
		\begin{array}{ll}
		\bm{U}_L &\text{if } x < {-c}(\varepsilon_L^*)\, t\, ,\\
		{\bm V}_1(x/t) &\text{if } {-c}(\varepsilon_L^*)\, t \leqslant x \leqslant {-c}(\varepsilon_M)\, t\, ,\\
		\bm{U}_M &\text{if } {-c}(\varepsilon_M)\, t \leqslant x \leqslant c(\varepsilon_M)\, t\, ,\\
		{\bm V}_2(x/t) &\text{if } c(\varepsilon_M)\, t \leqslant x < c(\varepsilon_R^*)\, t\, ,\\
		\bm{U}_R &\text{if } c(\varepsilon_R^*)\, t \leqslant x\, .\\
		\end{array}}\!\right.
	\end{equation}
	Here, $\varepsilon_M\approx 1.604\times 10^{-4}$. The rarefactions break at $\varepsilon_L^*=0$ and $\varepsilon_R^*=0.5\times 10^{-4}$ (\ref{EqWendroffLandau}).

	%-------------------------- figures --------------------------------
	
	\begin{figure}
		\begin{minipage}{0.25\linewidth}
			\centering
			(a)
			\vspace{0.5em}
			
			\includegraphics{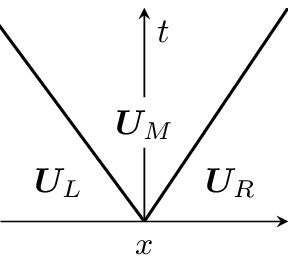}
			
			\vspace{1em}
			(b)
			\vspace{0.5em}
			
			\includegraphics{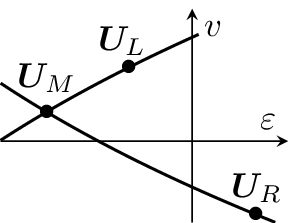}
		\end{minipage}
		\begin{minipage}{0.75\linewidth}
			\centering
			(c)
			
			~\hspace{0.5em}\includegraphics{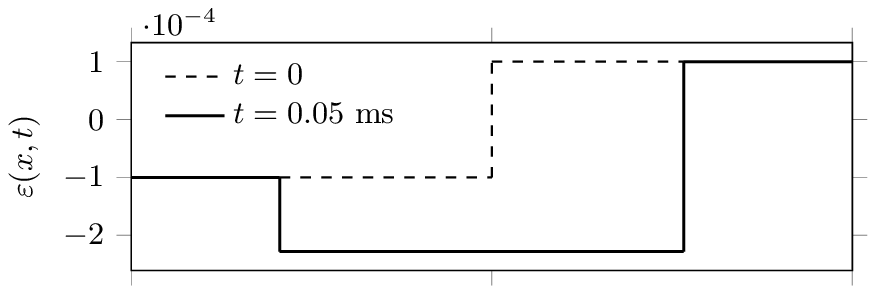}
			
			\vspace{-1em}
			\includegraphics{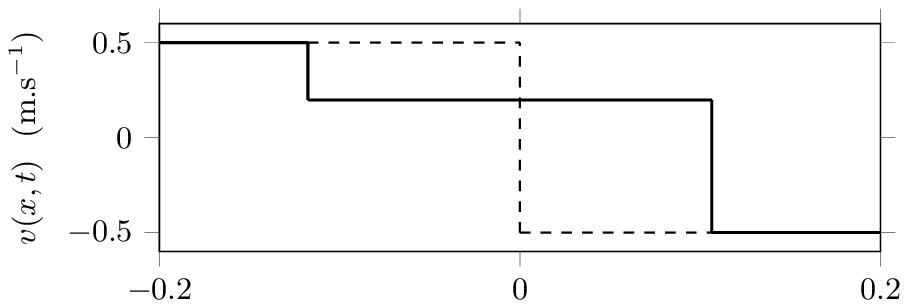}
		\end{minipage}
		\caption{(a) Solution to the Riemann problem for the hyperbola (\ref{BLawHyp}) with two shock waves. (b) Hugoniot loci. (c) Analytical solution at $t=0$ and $t=0.05$~ms.\label{fig:Hyp1}}
	\end{figure}
	\begin{figure}
		\begin{minipage}{0.25\linewidth}
			\centering
			(a)
			\vspace{0.5em}
			
			\includegraphics{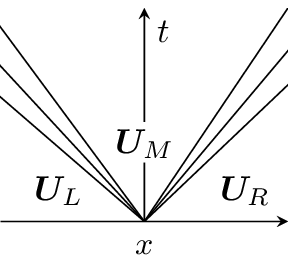}
			
			\vspace{1em}
			(b)
			\vspace{0.5em}
			
			\includegraphics{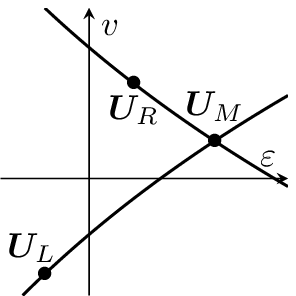}
		\end{minipage}
		\begin{minipage}{0.75\linewidth}
			\centering
			(c)
			
			~\hspace{1.2em}\includegraphics{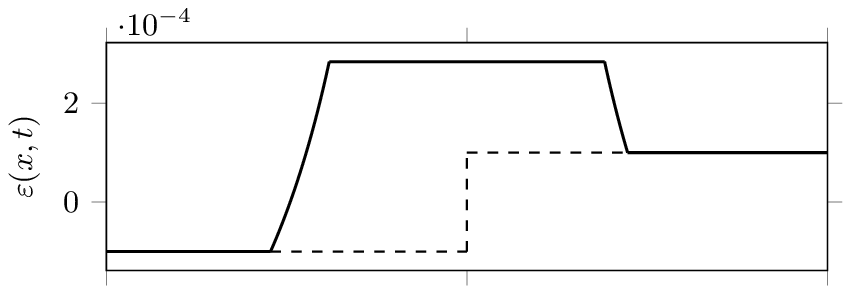}
			
			\vspace{-1em}	
			\includegraphics{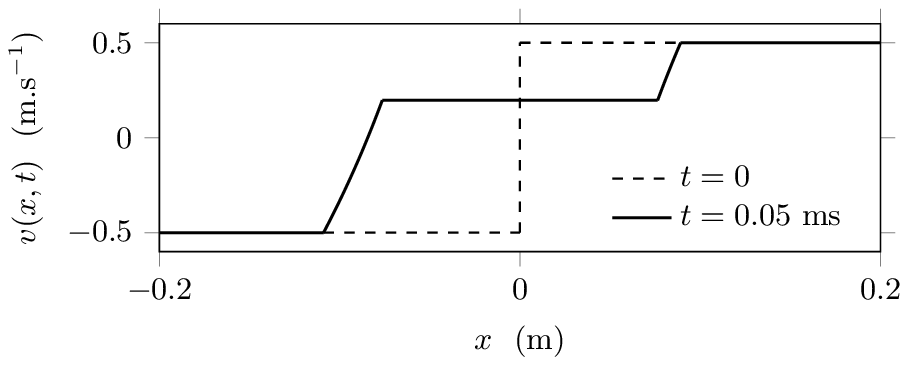}
		\end{minipage}
		\caption{(a) Solution to the Riemann problem for the hyperbola (\ref{BLawHyp}) with two rarefactions. (b) Rarefaction curves. (c) Analytical solution at $t=0$ and $t=0.05$~ms.\label{fig:Hyp4}}
	\end{figure}
	\begin{figure}
		\begin{minipage}{0.25\linewidth}
			\centering
			(a)
			\vspace{0.5em}
			
			\includegraphics{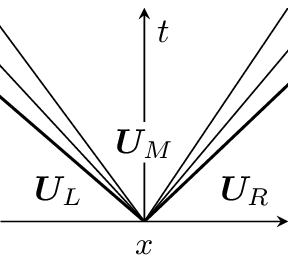}
			
			\vspace{1em}
			(b)
			\vspace{0.5em}
			
			\includegraphics{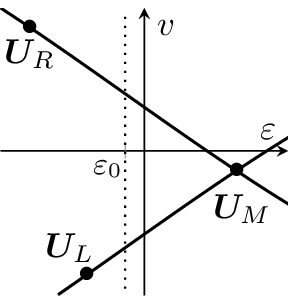}
		\end{minipage}
		\begin{minipage}{0.75\linewidth}
			\centering
			(c)
			
			~\hspace{-0.3em}\includegraphics{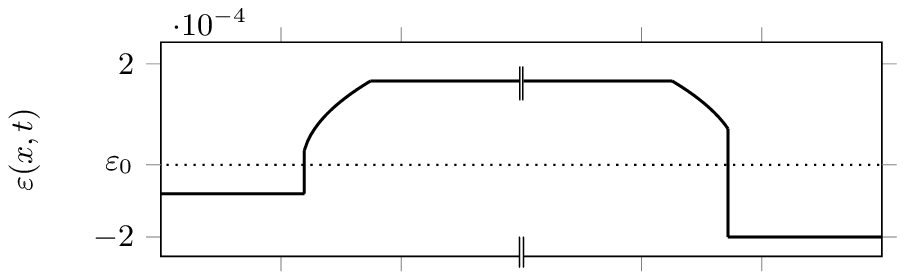}
			
			\vspace{-1em}	
			\includegraphics{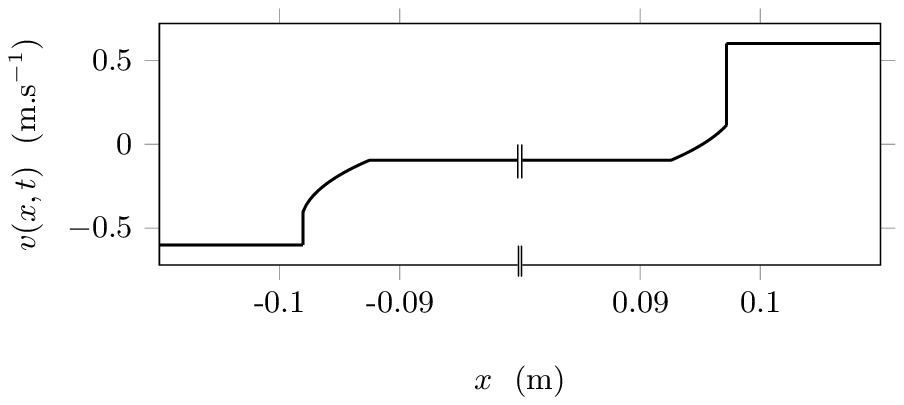}
		\end{minipage}
		\caption{(a) Solution to the Riemann problem for Landau's law (\ref{BLawLandau}) with two compound waves. (b) 1-shock-rarefaction and 2-rarefaction-shock curves. (c) Analytical solution at $t=0.05$~ms. The $x$-axis is broken from ${-0.08}$ to $0.08$~m.\label{fig:Lan1}}
	\end{figure}
	
	%%%%%%%%%%%%%%%%%%%%%%%%%%%%%%%%%%%%%%%%%%%%%%%%%%%%%%%%%%%%%%%%%%%
	%-------------------------- conclu --------------------------------
	
	\section{Conclusion}\label{sec:Conclusion}
	
	When the constitutive law is convex or concave, the system of 1D elastodynamics is similar to the $p$-system of barotropic gas dynamics. The $\varepsilon$-$v$ plane can be split into four admissibility regions: one for each combination of a 1-wave and a 2-wave \cite{godlewski96}. In this case, we obtain a new condition on the velocity jump $v_R - v_L$ which ensures the existence of the solution to the Riemann problem, whether the hyperbolicity domain is bounded or not. Also, we provide analytic expressions to compute the solution straightforwardly for the hyperbola and the quadratic Landau's law.
	
	These results have been extended to constitutive laws which are neither convex nor concave. Indeed, for constitutive laws with one inflection point, we obtain a new condition on the velocity jump which ensures the existence of the solution to the Riemann problem. Furthermore, we propose a partition of the $\varepsilon$-$v$ plane into nine admissibility regions. An application and a Matlab toolbox are freely available at {\color{blue}\url{http://gchiavassa.perso.centrale-marseille.fr/RiemannElasto/}}. The mathematics and the approach presented here could be applied to more complicated constitutive laws, e.g. with a disjoint union of inflexion points.
	
	%-------------------------- conclu --------------------------------
	
	\section{Acknowledgments}\label{sec:Aknow}
	
	We acknowledge St{\'e}phane Junca (JAD, Nice) for his bibliographical insights.
	
	%%%%%%%%%%%%%%%%%%%%%%%%%%%%%%%%%%%%%%%%%%%%%%%%%%%%%%%%%%%%%%%%%%%%
	%-------------------------- biblio --------------------------------

	%%%%%%%%%%%%%%%%%%%%%%%%%%%%%%%%%%%%%%%%%%%%%%%%%%%%%%%%%%%%%%%%%%%%%
	%-------------------------- appendix --------------------------------
	
	\appendix
	
	%-------------------------- details --------------------------------
	
	\section{}\label{sec:Maths}
	
	\subsection{Elementary wave curves}\label{subsec:WaveCurves}
	
	Here, we list some properties of the curves $\mathcal{S}_p^\ell$, $\mathcal{R}_p^\ell$, $\mathcal{SR}_p^\ell$ and $\mathcal{RS}_p^\ell$.
	
	\paragraph*{Discontinuities.}
	Let us differentiate equation (\ref{EqRH}). We obtain
	\begin{equation}
	{\addtolength{\jot}{0.5em}
		\begin{aligned}
		\frac{d}{d \varepsilon} \mathcal{S}_1^\ell(\varepsilon) &= \frac{1}{2} \sqrt{\frac{\sigma(\varepsilon) - \sigma(\varepsilon_\ell)}{\rho_0\, (\varepsilon - \varepsilon_\ell)}} \left( 1 + \sigma'(\varepsilon) \bigg/ \frac{\sigma(\varepsilon) - \sigma(\varepsilon_\ell)}{\varepsilon - \varepsilon_\ell} \right) = -\frac{d}{d \varepsilon} \mathcal{S}_2^\ell(\varepsilon)\\
		& > 0\, .
		\end{aligned}}
	\label{EqRHDiff}
	\end{equation}
	Therefore, $\mathcal{S}_1^\ell$ is an increasing bijection and $\mathcal{S}_2^\ell$ is a decreasing bijection.
	
	\paragraph*{Rarefactions.}
	Since $C$ is the primitive of a strictly positive continuous function, $C$ is strictly increasing and continuous. Therefore, $\mathcal{R}_1^\ell$ is an increasing bijection and $\mathcal{R}_2^\ell$ is a decreasing bijection (\ref{EqRarefaction}).
	
	\paragraph*{Shock-rarefactions.}
	Shock-rarefaction curves (\ref{EqSR}) have the same properties as rarefaction curves (\ref{EqRarefaction}). Indeed, they differ only by a constant, which equals zero if $\varepsilon_\ell=\varepsilon_\ell^*=\varepsilon_0$.
	
	\paragraph*{Rarefaction-shocks.}
	Let use differentiate equation (\ref{EqRS}). We obtain
	\begin{equation}
		\frac{d}{d \varepsilon} \mathcal{RS}_1^\ell(\varepsilon) = c(\varepsilon^*) - \frac{d\varepsilon^*}{d\varepsilon}\, c'(\varepsilon^*) \left(\varepsilon^* - \varepsilon\right) = -\frac{d}{d \varepsilon}  \mathcal{RS}_2^\ell(\varepsilon)\, ,
		\label{EqRSDiff1}
	\end{equation}
	where
	\begin{equation}
		c(\varepsilon^*)=\sqrt{\frac{\sigma'(\varepsilon^*)}{\rho_0}}
		\qquad\text{and}\qquad
		c'(\varepsilon^*)=\frac{\sigma''(\varepsilon^*)}{2 \sqrt{\rho_0\, \sigma'(\varepsilon^*)}}		 
		\, .
		\label{EqRSDiff0}
	\end{equation}
	Applying the implicit functions theorem to $F(\varepsilon^*,\varepsilon)$ in (\ref{EqWendroff}) requires $\partial F / \partial a (\varepsilon^*,\varepsilon) \neq 0$. Since $\partial F / \partial a (\varepsilon^*,\varepsilon) = \sigma''(\varepsilon^*)$, the hypotheses of the theorem are satisfied if $\varepsilon^*\neq \varepsilon_0$ (\ref{NLNGen}). Finally,
	\begin{equation}
	{\addtolength{\jot}{0.5em}
		\begin{aligned}
		\frac{d\varepsilon^*}{d\varepsilon} &= - \frac{\partial F / \partial b}{\partial F / \partial a} (\varepsilon^*,\varepsilon)
		\\
		&= \frac{\sigma'(\varepsilon^*) - \sigma'(\varepsilon)}{\sigma''(\varepsilon^*)(\varepsilon^* - \varepsilon)}\, .
		\end{aligned}}
	\label{Implicit}
	\end{equation}
	Thus,
	\begin{equation}
	{\addtolength{\jot}{0.2em}
		\begin{aligned}
		\frac{d}{d \varepsilon} \mathcal{RS}_1^\ell(\varepsilon) &= \frac{\sigma'(\varepsilon^*) + \sigma'(\varepsilon)}{2 \sqrt{\rho_0\, \sigma'(\varepsilon^*)}} = -\frac{d}{d \varepsilon} \mathcal{RS}_2^\ell(\varepsilon)\\
		&> 0\, .
		\end{aligned}}
	\label{EqRSDiff2}
	\end{equation}
	Therefore, $\mathcal{RS}_1^\ell$ is an increasing bijection and $\mathcal{RS}_2^\ell$ is a decreasing bijection.
	
	%---------------------------------------------------------
	
	\subsection{Restriction on the velocity jump}\label{subsec:MathsJump}
	
	In this section, we provide analytical expressions deduced from theorems~\ref{thm:Intersect} and \ref{thm:IntersectNLNGen}.
	
	\paragraph*{Concave constitutive laws.}
	We go back to the condition that must be satisfied by the initial data when the constitutive law is concave, i.e. equation (\ref{IntersectCond}) in theorem~\ref{thm:Intersect}. According to the expressions of $\Phi_1^L$ and $\Psi_2^R$ in (\ref{WNLGen}), one has
	\begin{equation}
		\underset{\varepsilon\rightarrow \varepsilon_\textit{inf}+}{\lim} \mathcal{S}_2^R(\varepsilon) - \mathcal{S}_1^L(\varepsilon) > 0
		\qquad\text{and}\qquad
		\underset{\varepsilon\rightarrow \varepsilon_\textit{sup}-}{\lim} \mathcal{R}_2^R(\varepsilon) - \mathcal{R}_1^L(\varepsilon) < 0\, .
		\label{IntersectCondCW1}
	\end{equation}
	This can be expressed in terms of the velocity jump $v_R - v_L$. Based on (\ref{EqRH}) and (\ref{EqRarefaction}), condition (\ref{IntersectCondCW1}) becomes
	\begin{equation}
	\left\lbrace
	{\addtolength{\jot}{0.1em}
		\begin{aligned}
		& v_R - v_L > {-\underset{\varepsilon\rightarrow \varepsilon_\textit{inf}+}{\lim}} \!\left(\! \sqrt{\frac{\sigma(\varepsilon)-\sigma(\varepsilon_L)}{\rho_0} (\varepsilon-\varepsilon_L)} + \sqrt{\frac{\sigma(\varepsilon)-\sigma(\varepsilon_R)}{\rho_0} (\varepsilon-\varepsilon_R)} \right) ,\\
		& v_R - v_L < \underset{\varepsilon\rightarrow \varepsilon_\textit{sup}-}{\lim} 2\, C(\varepsilon) - C(\varepsilon_L) - C(\varepsilon_R) \, .
		\end{aligned}}
	\right.
	\label{IntersectCondNLGen}
	\end{equation}
	
	\paragraph*{Convex-concave constitutive laws.}
	The same condition (\ref{IntersectCondNLNGenThm}) must be satisfied by the initial data when the constitutive law is strictly convex for $\varepsilon<\varepsilon_0$ and strictly concave for $\varepsilon>\varepsilon_0$ (theorem~\ref{thm:IntersectNLNGen}). The expressions of $\Phi_1^L$ and $\Psi_2^R$ are given by (\ref{WNLNGen}). For instance, when $\varepsilon$ tends towards $\varepsilon_\textit{inf}$ in $\Phi_1^L(\varepsilon)$, one needs a comparison between $\varepsilon_L^*$ and $\varepsilon_\textit{inf}$ to choose the correct elementary wave curve. Since $\sigma'(\varepsilon_L^*)>0=\sigma'(\varepsilon_\textit{inf})$, it is immediate that $\varepsilon_L^*>\varepsilon_\textit{inf}$. Similar comparisons can be written to select the correct elementary curve in $\Psi_2^R(\varepsilon)$ or when $\varepsilon$ tends towards $\varepsilon_\textit{sup}$. Finally, (\ref{IntersectCondNLNGenThm}) writes
	\begin{itemize}
		\item if $\varepsilon_L \geqslant \varepsilon_0$ and $\varepsilon_R \geqslant \varepsilon_0$
		\begin{equation}
			\underset{\varepsilon\rightarrow \varepsilon_\textit{inf}+}{\lim} \mathcal{SR}_2^R(\varepsilon) - \mathcal{SR}_1^L(\varepsilon) > 0
			\quad\text{and}\quad
			\underset{\varepsilon\rightarrow \varepsilon_\textit{sup}-}{\lim} \mathcal{R}_1^L(\varepsilon) - \mathcal{R}_2^R(\varepsilon) > 0 \, ,
			\label{IntersectCondNLNGen1}
		\end{equation}
		\item if $\varepsilon_L \geqslant \varepsilon_0 > \varepsilon_R$
		\begin{equation}
			\underset{\varepsilon\rightarrow \varepsilon_\textit{inf}+}{\lim} \mathcal{R}_2^R(\varepsilon) - \mathcal{SR}_1^L(\varepsilon) > 0
			\quad\text{and}\quad
			\underset{\varepsilon\rightarrow \varepsilon_\textit{sup}-}{\lim} \mathcal{R}_1^L(\varepsilon) - \mathcal{SR}_2^R(\varepsilon) > 0 \, ,
			\label{IntersectCondNLNGen2}
		\end{equation}
		\item if $\varepsilon_R \geqslant \varepsilon_0 > \varepsilon_L$
		\begin{equation}
			\underset{\varepsilon\rightarrow \varepsilon_\textit{inf}+}{\lim} \mathcal{SR}_2^R(\varepsilon) - \mathcal{R}_1^L(\varepsilon) > 0
			\quad\text{and}\quad
			\underset{\varepsilon\rightarrow \varepsilon_\textit{sup}-}{\lim} \mathcal{SR}_1^L(\varepsilon) - \mathcal{R}_2^R(\varepsilon) > 0 \, ,
			\label{IntersectCondNLNGen3}
		\end{equation}
		\item if $\varepsilon_L<\varepsilon_0$ and $\varepsilon_R<\varepsilon_0$
		\begin{equation}
			\underset{\varepsilon\rightarrow \varepsilon_\textit{inf}+}{\lim} \mathcal{R}_2^R(\varepsilon) - \mathcal{R}_1^L(\varepsilon) > 0
			\quad\text{and}\quad
			\underset{\varepsilon\rightarrow \varepsilon_\textit{sup}-}{\lim} \mathcal{SR}_1^L(\varepsilon) - \mathcal{SR}_2^R(\varepsilon) > 0 \, .
			\label{IntersectCondNLNGen4}
		\end{equation}
	\end{itemize}
	Based on the expressions of the elementary wave curves (\ref{EqRH}), (\ref{EqRarefaction}), (\ref{EqRS}) and (\ref{EqSR}), inequalities (\ref{IntersectCondNLNGen1})-(\ref{IntersectCondNLNGen4}) become
	\begin{itemize}
		\item if $\varepsilon_L \geqslant \varepsilon_0$ and $\varepsilon_R \geqslant \varepsilon_0$
	\end{itemize}
	\begin{equation}
	\left\lbrace
	{\addtolength{\jot}{0.1em}
		\begin{aligned}
		v_R - v_L > &\underset{\varepsilon\rightarrow \varepsilon_\textit{inf}+}{\lim} 2\,C(\varepsilon) - C(\varepsilon_L^*) - c(\varepsilon_L^*)(\varepsilon_L-\varepsilon_L^*) \\
		&- C(\varepsilon_R^*) - c(\varepsilon_R^*)(\varepsilon_R-\varepsilon_R^*)\, ,\\
		v_R - v_L < &\underset{\varepsilon\rightarrow \varepsilon_\textit{sup}-}{\lim} 2\, C(\varepsilon) - C(\varepsilon_L) - C(\varepsilon_R)\, ,
		\end{aligned}}
	\right.
	\label{IntersectCondNLNGenCalc1}
	\end{equation}
	\begin{itemize}
		\item if $\varepsilon_L\geqslant\varepsilon_0>\varepsilon_R$
	\end{itemize}
	\begin{equation}
	\left\lbrace
	{\addtolength{\jot}{0.1em}
		\begin{aligned}
		v_R - v_L > &\underset{\varepsilon\rightarrow \varepsilon_\textit{inf}+}{\lim} 2\,C(\varepsilon) - C(\varepsilon_L^*) - c(\varepsilon_L^*)(\varepsilon_L-\varepsilon_L^*) - C(\varepsilon_R)\, ,\\
		v_R - v_L < &\underset{\varepsilon\rightarrow \varepsilon_\textit{sup}-}{\lim} 2\,C(\varepsilon) - C(\varepsilon_L) - C(\varepsilon_R^*) - c(\varepsilon_R^*)(\varepsilon_R-\varepsilon_R^*)\, ,
		\end{aligned}}
	\right.
	\label{IntersectCondNLNGenCalc2}
	\end{equation}
	\begin{itemize}
		\item if $\varepsilon_R\geqslant\varepsilon_0>\varepsilon_L$
	\end{itemize}
	\begin{equation}
	\left\lbrace
	{\addtolength{\jot}{0.1em}
		\begin{aligned}
		v_R - v_L > &\underset{\varepsilon\rightarrow \varepsilon_\textit{inf}+}{\lim} 2\,C(\varepsilon) - C(\varepsilon_L) - C(\varepsilon_R^*) - c(\varepsilon_R^*)(\varepsilon_R-\varepsilon_R^*)\, ,\\
		v_R - v_L < &\underset{\varepsilon\rightarrow \varepsilon_\textit{sup}-}{\lim} 2\,C(\varepsilon) - C(\varepsilon_L^*) - c(\varepsilon_L^*)(\varepsilon_L-\varepsilon_L^*) - C(\varepsilon_R)\, ,
		\end{aligned}}
	\right.
	\label{IntersectCondNLNGenCalc3}
	\end{equation}
	\begin{itemize}
		\item if $\varepsilon_L<\varepsilon_0$ and $\varepsilon_R<\varepsilon_0$
	\end{itemize}
	\begin{equation}
	\left\lbrace
	{\addtolength{\jot}{0.1em}
		\begin{aligned}
		v_R - v_L > &\underset{\varepsilon\rightarrow \varepsilon_\textit{inf}+}{\lim} 2\, C(\varepsilon) - C(\varepsilon_L) - C(\varepsilon_R)\, ,\\
		v_R - v_L < &\underset{\varepsilon\rightarrow \varepsilon_\textit{sup}-}{\lim} 2\,C(\varepsilon) - C(\varepsilon_L^*) - c(\varepsilon_L^*)(\varepsilon_L-\varepsilon_L^*) \\
		& - C(\varepsilon_R^*) - c(\varepsilon_R^*)(\varepsilon_R-\varepsilon_R^*)\, .
		\end{aligned}}
	\right.
	\label{IntersectCondNLNGenCalc4}
	\end{equation}

\end{document}